\documentclass[twoside]{IEEEtran}

\usepackage{amsmath,amsfonts,amssymb,amscd,mathrsfs,eucal,amsfonts,color}
\usepackage{graphicx}

\newtheorem{theorem}{Theorem}[section]
\newtheorem{corollary}[theorem]{Corollary}
\newtheorem{lemma}[theorem]{Lemma}
\newtheorem{proposition}[theorem]{Proposition}
\newtheorem{definition}[theorem]{Definition}
\newtheorem{remark}[theorem]{Remark}

\def\r{\rho}
\def\d{\delta}
\def\e{\epsilon}
\def\g{\gamma}
\def\G{\Gamma}
\def\a{\alpha}

\title{Vanishingly Sparse Matrices and Expander Graphs, \\With Application to Compressed Sensing}

\author{Bubacarr ~Bah, and ~Jared ~Tanner
\thanks{Laboratory for Information and Inference Systems (LIONS), \'Ecole Polytechnique F\'ed\'erale de Lausanne (EPFL), Lausanne, Switzerland ({\tt bubacarr.bah@epfl.ch}).}
\thanks{Mathematics Institute and Exeter College, University of Oxford, Oxford, UK ({\tt tanner@maths.ox.ac.uk}). J.T. acknowledges support from the Leverhulme Trust.}
\thanks{Copyright (c) 2012 IEEE}}

\begin{document}
\maketitle
\IEEEpeerreviewmaketitle

\begin{abstract}
	We revisit the probabilistic construction of sparse random matrices where each column has a fixed number of nonzeros whose row indices are drawn uniformly at random with replacement. These matrices have a one-to-one correspondence with the adjacency matrices of fixed left degree expander graphs.  We present formulae for the expected cardinality of the {\em set of neighbors} for these graphs, and present tail bounds on the probability that this cardinality will be less than the expected value. Deducible from these bounds are similar bounds for the {\em expansion} of the graph which is of interest in many applications. These bounds are derived through a more detailed analysis of collisions in unions of sets. Key to this analysis is a novel {\em dyadic splitting} technique. The analysis led to the derivation of better order constants that allow for quantitative theorems on existence of lossless expander graphs and hence the sparse random matrices we consider and also quantitative compressed sensing sampling theorems when using sparse non mean-zero measurement matrices.
\end{abstract}


\section{Introduction}\label{sec:intro}
Sparse matrices are particularly useful in applied and computational mathematics because of their low storage complexity and fast implementation as compared to dense matrices, see  \cite{horn1990matrix,demmel1993numerical,trefethen1997numerical}. Of late, significant progress has been made to incorporate sparse matrices in compressed sensing, with \cite{xu2007further,berinde2008combining,berinde2008sparse,jafarpour2009efficient} giving both theoretical performance guarantees and also exhibiting numerical results that shows sparse matrices coming from expander graphs can be as good sensing matrices as their dense counterparts. In fact, Blanchard and Tanner \cite{blanchard2012gpu}  recently demonstrated in a GPU implementation how well these type of matrices do compared to dense Gaussian and Discrete Cosine Transform matrices even with very small fixed number of nonzeros per column (as considered here).

In this manuscript we consider random sparse matrices that are adjacency matrices of lossless expander graphs. Expander graphs are highly connected graphs with very sparse adjacency matrices, a precise definition of a lossless expander graph is given in Definition \ref{def:expander} and their illustration in Fig. \ref{fig:expander}.
\begin{definition}
\label{def:expander}
$G=\left(U,V,E\right)$ is a lossless $(k,d,\epsilon)$-expander if it is a bipartite graph with $|U| = N$ left vertices, $|V| = n$ right vertices and has a regular left degree $d$, such that any $X \subset U$ with $|X| \leq k$ has $|\Gamma(X)| = \left(1 - \epsilon \right) d |X|$ neighbors.
\footnote{{\em Lossless expanders} with parameters $d,k,n,N$ are equivalent to {\em lossless conductors} with parameters that are base 2 logarithms of the parameters of {\em lossless expanders} see \cite{hoory2006expander,berinde2008combining} and the references therein.}
\end{definition}
\begin{remark}
 \begin{enumerate}
  \item The graphs are  {\em lossless} because $\e\ll1$;
  \item They are called {\em unbalanced expanders} when $n\ll N$;
  \item The {\em expansion} of a lossless $(k,d,\e)$-expander graph is $\left(1 - \epsilon \right) d$.
 \end{enumerate}
\end{remark}

\begin{figure}[h]
\centering
\includegraphics[width=0.5\textwidth]{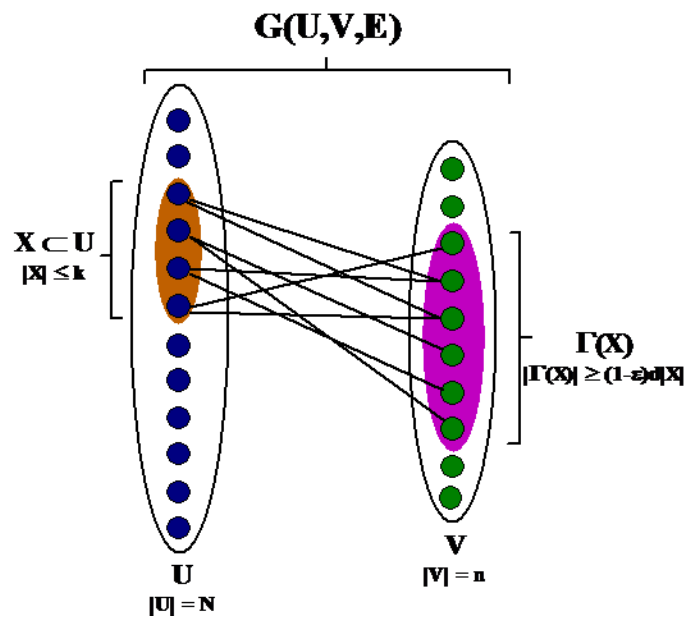}
\caption{Illustration of a lossless $(k,d,\e)$-expander graphs with $k=4$ and $d=2$.}
\label{fig:expander}
\end{figure}

Such graphs have been well studied in theoretical computer science and pure mathematics and have many applications including: Distributed Routing in Networks, Linear Time Decodable Error-Correcting Codes, Bitprobe Complexity of Storing Subsets, Fault-tolerance and a Distributed Storage Method, and Hard Tautologies in Proof Complexity, see \cite{capalbo2002randomness} or \cite{hoory2006expander} for a more detailed survey. Pinsker and Bassylago \cite{bassalygo1973complexity} proved the existence of lossless expanders and showed that any random left-regular bipartite graph is, with high probability, an expander graph. Probabilistic constructions with optimal parameters $n,N$ exist but are not suitable for the applications we consider here.  Deterministic constructions only achieve sub-optimal parameters, see Guruswami et. al. \cite{guruswami2007unbalanced}.

Our main contribution, is the presentation of quantitative guarantees on the probabilistic construction of these objects in the form of a bound on the tail probability of the size of the {\em set of neighbors}, $\Gamma(X)$ for a given $X \subset U$, of a randomly generated left-degree bipartite graph. Moreover, we provide deducible bounds on the tail probability of the {\em expansion} of the graph, $|\G(X)|/|X|$. We derive quantitative guarantees for randomly generated non-mean zero sparse binary matrices to be adjacency matrices of expander graphs. In addition, we derive the first phase transitions showing regions in parameter space that depict when a left-regular bipartite graph with a given set of parameters is guaranteed to be a lossless expander with high probability. The most significant contribution, which is the key innovation, of this paper is the use of a novel technique of {\em dyadic splitting of sets}. We derived our bounds using this technique and apply them to derive $\ell_1$ restricted isometry constants ($\mathrm{RIC}_1$).

Numerous compressed sensing algorithms have been designed for sparse matrices \cite{xu2007further,berinde2008combining,berinde2008sparse,jafarpour2009efficient}. Another contribution of our work is the derivation of sampling theorems, presented as phase transitions, comparing performance guarantees for some of these algorithms as well as the more traditional $\ell_1$ minimization compressed sensing formulation. We also show how favorably $\ell_1$ minimization performance guarantees for such sparse matrices compared to what $\ell_2$ restricted isometry constants ($\mathrm{RIC}_2$) analysis yields for the dense Gaussian matrices. For this comparison, we used sampling theorems and phase transitions from related work by Blanchard et. al. \cite{blanchard2011phase} that provided such theorems for dense Gaussian matrices based on $\mathrm{RIC}_2$ analysis.

The outline of the rest of this introduction section goes as follows. In Section \ref{sec:main} we present our main results in Theorem \ref{thm:prob_bound_1set_expansion} and Corollary \ref{cor:prob_bound_1set_expander}. In Section \ref{sec:ric1} we discuss $\mathrm{RIC}_1$ and its implication for compressed sensing, leading to two sampling theorems in Corollaries \ref{cor:prob_expander_existence1} and \ref{cor:prob_expander_existence2}.

\subsection{Main results}\label{sec:main}
Our main results is about a class of sparse matrices coming from lossless expander graphs, a class which include non-mean zero matrices. We start by defining the class of matrices we consider and a key concept of a {\em set of neighbors} used in the derivation of the main results of the manuscript. Firstly, we denote $\hbox{H}(p) := -p\log(p) - (1-p)\log(1-p)$ as the Shannon entropy function of base $e$ logarithm.

\begin{definition}
 \label{def:1set_expansion}
Let $A$ be an $n\times N$ matrix with $d$ nonzeros in each column. We refer to $A$ as a random
\begin{enumerate}
 \item sparse expander (SE) if every nonzero has value $1$
 \item sparse signed expander (SSE) if every nonzero has value from $\{-1,1\}$
\end{enumerate}
and the support set of the $d$ nonzeros per column are drawn uniformly at random, with each column drawn independently.
\end{definition}

SE matrices are adjacency matrices of lossless $(k,d,\e)$-expander graphs while SSE matrices have random sign patterns in the nonzeros of an adjacency matrix of a lossless $(k,d,\e)$-expander graph. If $A$ is either an SE or SSE it will have only $d$ nonzeros per column and since we fix $d\ll n$, $A$ is therefore ``vanishingly sparse.'' We denote $A_S$ as a submatrix of $A$ composed of columns of $A$ indexed by the set $S$ with $|S|=s$.  To aid translation between the terminology of graph theory and linear algebra we define the {\em set of neighbors} in both notation.

\begin{definition}
 \label{def:neighbours}
Consider a bipartite graph $G=(U,V,E)$ where $E$ is the set of edges and $e_{ij}=(x_i,y_j)$ is the edge that connects vertex $x_i$ to vertex $y_j$. For a set of left vertices $S\subset U$ its set of neighbors is $\Gamma(S) = \{y_j|x_i\in S \mbox{ and } e_{ij}\in E\}$. In terms of the adjacency matrix, $A$, of $G=(U,V,E)$ the set of neighbors of $A_S$ for $|S|=s$, denoted by $A_s$, is the set of rows with at least one nonzero.
\end{definition}

\begin{definition}
 \label{def:expansion}
Using Definition \ref{def:neighbours} the expansion of the graph is given by the ratio $|\G(S)|/|S|$, or equivalently, $|A_s|/s$.
\end{definition}

By the definition of a lossless expander, Definition \ref{def:expander}, we need $\left|\Gamma(S)\right|$ to be large for every small $S\subset U$. In terms of the class of matrices defined by Definition \ref{def:1set_expansion}, for every $A_S$ we want to have $|A_s|$ as close to $n$ as possible, where $n$ is the number of rows. Henceforth, we will only use the linear algebra notation $A_s$ which is equivalent to $\Gamma(S)$.
Note that $\left|A_s\right|$ is a random variable depending on the draw of the set of columns, $S$, for each fixed $A$. Therefore, we can ask what is the probability that $\left|A_s\right|$ is not greater than $a_s$, in particular where $a_s$ is smaller than the expected value of $\left|A_s\right|$. This is the question that Theorem \ref{thm:prob_bound_1set_expansion} to answers. We then use this theorem with $\mathrm{RIC}_1$ to deduce the corollaries that follow which are about the probabilistic construction of expander graphs, the matrices we consider, and sampling theorems of some selected compressed sensing algorithms.

\begin{theorem}
\label{thm:prob_bound_1set_expansion}
For fixed $s,n,N$ and $d$, let an $n\times N$ matrix, $A$ be drawn from either of the classes of matrices defined in Definition \ref{def:1set_expansion}, then
\begin{multline}
\label{eq:prob_bound_1set_expansion}
 \hbox{Prob}\left(\left|A_s\right| \leq a_s\right) < p_{max}(s,d) \\ \times \exp\left[n\cdot\Psi\left(a_s,\ldots,a_1\right)\right]
\end{multline}
where $p_{max}(s,d)$ is given by
\begin{equation}
 \label{eq:pmax}
    p_{max}(s,d) = \frac{2}{25\sqrt{2\pi s^3d^3}}, \qquad \mbox{and}
\end{equation}
\begin{multline}
 \label{eq:Psi}
    \Psi\left(a_s,\ldots,a_1\right) = \frac{1}{n} \Bigg{[}\sum_{i=1}^{\lceil s/2\rceil} \frac{s}{2i}\left(\left(n-a_i\right) \cdot \hbox{H}\left(\frac{a_{2i}-a_i}{n-a_i}\right)\right. \\
    \left. + a_i\cdot \hbox{H}\left(\frac{a_{2i}-a_i}{a_i}\right) - n\cdot \hbox{H}\left(\frac{a_i}{n}\right) \right) + 3s\log\left(5d \right) \Bigg{]}
\end{multline}
where $a_1:=d$.

If no restriction is imposed on $a_s$ then the $a_i$ for $i>1$ take on their expected value $\widehat{a}_{i}$ given by
\begin{align}
\label{eq:aj_system_unconstrained}
\widehat{a}_{2i} = \widehat{a}_{i}\left(2 - \frac{\widehat{a}_{i}}{n}\right) \quad \mbox{for} \quad i=1,2,4,\ldots,\lceil s/2\rceil.
\end{align}

If $a_{s}$ is restricted to be less than $\widehat{a}_{s}$, then the $a_i$ for $i>1$ are the unique solutions to the following polynomial system
\begin{equation}
\label{eq:aj_system_constrained}
a_{2i}^3 - 2a_ia_{2i}^2 + 2a_i^2a_{2i} - a_i^2a_{4i} = 0 ~ \mbox{for } ~ i = 1, 2,\ldots,\lceil s/4\rceil
\end{equation}
with $a_{2i}\ge a_i$ for each $i$.
\end{theorem}

\begin{corollary}
\label{cor:prob_bound_1set_expander}
For fixed $s,n,N,d$ and $0<\e<1/2$, let an $n\times N$ matrix, $A$ be drawn from the class of matrices defined in Definition \ref{def:1set_expansion}, then
\begin{multline}
\label{eq:prob_bound_1set_expander}
 \hbox{Prob}\left(\mathop{\|A_Sx\|_1} \leq (1-2\e)d\|x\|_1\right) < p_{max}(s,d)\\ \times \exp\left[n\cdot\Psi\left(s,d,\e\right)\right]
\end{multline}
where $\Psi\left(s,d,\e\right) = \Psi\left(a_s,\ldots,a_1\right)$ in \eqref{eq:Psi} with $a_s = (1-\e)ds$ and $a_1 = d$, and $p_{max}(s,d)$ is the polynomial in \eqref{eq:pmax}.
\end{corollary}

Theorem \ref{thm:prob_bound_1set_expansion} and Corollary \ref{cor:prob_bound_1set_expander} allow us to calculate $s,n,N,d,\e$ where
the probability of the probabilistic constructions in Definition \ref{def:1set_expansion} not being a lossless $(s,d,\e)$-expander is exponentially small.  For moderate values of $\e$ this allows us to make quantitative sampling theorems for some compressed sensing reconstruction algorithms.

\subsection{$\mathrm{RIC}_1$ and its implications to Compressed Sensing}\label{sec:ric1}
In compressed sensing, and by extension in sparse approximation, we observe the effect of the application of a matrix to a vector of interest and we endeavor to recovery this vector of interest by exploiting the inherent simplicity in this vector. Precisely, let $x \in \mathbb{R}^N,$ be the vector of interest whose simplicity is that it has $k<N$ nonzeros, which we refer to as $k-$sparse; then we observe $y \in \mathbb{R}^n,$ as the measurement vector resulting from the multiplication of $x$ by an ${n\times N}$ matrix, $A$. The minimum simplicity reconstruct of $x$ can be written as
\begin{equation}
\label{eq:decoding_exact}
 \min_{x \in \chi^N} \|x\|_0 \quad \mbox{subject to} \quad Ax = y,
\end{equation}
where $\chi^N$ is the set of all $k-$sparse vectors and $\|z\|_0$ counts the nonzero components of $z$; this model may be reformulated to include noise in the measurements. References \cite{baranuik2007compressive,candes2008introduction,donoho2006compressed,lustig2008compressed} give detailed introductions to compressed sensing and its applications; while \cite{blumensath2009iterative,candes2005decoding,dai2009subspace,foucart2009sparsest,needell2009cosamp,berinde2008combining,berinde2008sparse,jafarpour2008efficient,xu2007further,xu2007efficient} provide information on some of the popular computationally efficient algorithms used to solve problem \eqref{eq:decoding_exact} and its reformulations.

We are able to give guarantees on the quality of the reconstructed vector from $A$ and $y$ from a variety of reconstruction algorithms. One of these guarantees is a bound on the approximation error between our recovered vector, say $\hat{x}$, and the original vector by the best $k$-term representation error i.e. $\|x-\hat{x}\|_1 \leq Const. \|x-x_k\|_1$ where $x_k$ is the optimal $k$-term representation for $x$. This is possible if $A$ has small $\mathrm{RIC}_1$, in other words $A$ satisfies the $\ell_1$ restricted isometry property (RIP-1), introduced by Berinde et. al. in \cite{berinde2008combining} and defined as thus.
\begin{definition}[RIP-1]
\label{def:ric1}
Let $\chi^N$ be the set of all $k-$sparse vectors, then an $n\times N$ matrix $A$ has RIP-1, with the lower $\mathrm{RIC}_1$ being the smallest $L(k,n,N;A)$, when the following condition holds.
\begin{equation}
\label{eq:ric1}
 \left(1-L(k,n,N;A)\right)||x||_1 \leq ||Ax||_1 \leq ||x||_1 ~~ \forall x \in \chi^N.
\end{equation}
\end{definition}

For computational purposes it is preferable to have $A$ sparse, but little quantitative information on $L(k,n,N;A)$ has been available for large sparse rectangular matrices. Berinde et. al. in \cite{berinde2008combining} showed that scaled adjacency matrices of lossless expander graphs (i.e. scaled SE matrices) satisfy RIP-1, and the same proof extends to the signed adjacency matrices (i.e. so called SSE matrices).
\begin{theorem}
\label{thm:ric1}
If an $n\times N$ matrix $A$ is either SE or SSE defined in Definition \ref{def:1set_expansion}, then $A/d$ satisfies RIP-1 with $L(k,n,N;A) = 2\e$.
\end{theorem}

\begin{proof}
The proof of the signed case (SSE) follows that of the unsigned case (SE) in \cite{berinde2008combining} but with absolute values included in the appropriate stages.
\end{proof}

Based on Theorem \ref{thm:ric1} which guarantees RIP-1, \eqref{eq:ric1}, for the class of matrices in Definition \ref{def:1set_expansion}, we give a bound, in Corollary \ref{cor:prob_expander_existence1}, for the probability that a random draw of a matrix with $d ~1$s or $\pm 1$s in each column fails to satisfy the lower bound of RIP-1 and hence fails to come from the class of matrices given in Definition \ref{def:1set_expansion}. In addition to Theorem \ref{thm:ric1}, Corollary \ref{cor:prob_expander_existence1} follows from Theorem \ref{thm:prob_bound_1set_expansion} and Corollary \ref{cor:prob_bound_1set_expander}.

\begin{corollary}
\label{cor:prob_expander_existence1}
Considering RIP-1, if $A$ is drawn from the class of matrices in Definition \ref{def:1set_expansion} and $0<\e<1/2$ with $k,n,N$ fixed, then for all $k$-sparse vectors $x$
\begin{multline}
\label{eq:prob_expander_existence1}
\hbox{Prob}\left(\|Ax\|_1 \leq (1-2\e)d\|x\|_1\right) < p'_{max}(N,k,d) \\ \times \exp\left[N\cdot \Psi_{net}\left(k,n,N;d,\e\right)\right]
\end{multline}
where $p'_{max}(N,k,d)$ and $\Psi_{net}$ are given by
\begin{align}
 \label{eq:ppmax}
    p'_{max}(N,k,d) & = \frac{1}{16\pi k\sqrt{d^3\left(1-\frac{k}{N}\right)}}, \\
 \label{eq:Psi_net}
    \Psi_{net}\left(k,n,N;d,\e\right) & =  \hbox{H}\left(\frac{k}{N}\right) + \frac{n}{N}\Psi\left(k,d,\e\right),
\end{align}
with $\Psi\left(k,d,\e\right)$ defined in Corollary \ref{cor:prob_bound_1set_expander}.
\end{corollary}

Furthermore, the following corollary is a consequence of Corollary \ref{cor:prob_expander_existence1} and it is a sampling theorem on the existence of lossless expander graphs. The proof of Corollaries \ref{cor:prob_expander_existence1} and \ref{cor:prob_expander_existence2} are presented in Sections \ref{sec:cor_prob_expander_exist1} and \ref{sec:cor_prob_expander_exist2} respectively.
\begin{corollary}
\label{cor:prob_expander_existence2}
Consider $0<\e<1/2$ and $d$ fixed.   If $A$ is drawn from the class of matrices in Definition \ref{def:1set_expansion} and all $x$ drawn from $\chi^N$ with $(k,n,N) \rightarrow \infty$ while $k/n \rightarrow \r \in (0,1)$ and $n/N \rightarrow \d \in (0,1)$ then for $\r < (1-\g)\r^{exp}(\d;d,\e)$ and $\g>0$
\begin{equation}
\label{eq:prob_expander_existence2}
\hbox{Prob}\left(\|Ax\|_1 \geq (1-2\e)d\|x\|_1\right) \rightarrow 1
\end{equation}
exponentially in $n$, where $\r^{exp}(\d;d,\e)$ is the largest limiting value of $k/n$ for which
\begin{equation}
 \label{eq:rho_exp}
    \hbox{H}\left(\frac{k}{N}\right) + \frac{n}{N}\Psi\left(k,d,\e\right) = 0.
\end{equation}
\end{corollary}

The outline of the rest of the manuscript is as follows:  In Section \ref{sec:dd_main_results} we show empirical data to validate our main results and also present lemmas (and their proofs) that are key to the proof of the main theorem, Theorem \ref{thm:prob_bound_1set_expansion}. In Section \ref{sec:ric12} we discuss restricted isometry constants and compressed sensing algorithms. In Section \ref{sec:proof} we prove the mains results, that is Theorem \ref{thm:prob_bound_1set_expansion} and the corollaries in Sections \ref{sec:main} and \ref{sec:ric1}. Section \ref{sec:appendix} is the appendix where we present the alternative to Theorem \ref{thm:prob_bound_1set_expansion}.

\section{Discussion and derivation of the main results}\label{sec:dd_main_results}
We present the method used to derive the main results and discuss the validity and implications of the method. We start by presenting in the next subsection, Section \ref{sec:numerics}, numerical results that support the claims of the main results in Sections \ref{sec:main} and \ref{sec:ric1}. This is followed in Section \ref{sec:dd_main_results} with lemmas, propositions and corollaries and their proofs.

\subsection{Discussion on main results}\label{sec:numerics}
Theorem  \ref{thm:prob_bound_1set_expansion} gives a bound on the probability that the cardinality of a union of $k$ sets each with $d$ elements, i.e. $|A_k|$, is less than $a_k$. Fig. \ref{fig:mean_neighbour} shows plots of values of $|A_k|$ (size of set of neighbors) for different $k$ (in blue), superimposed on these plots is the mean value of $|A_k|$ (in red) both taken over 500 realizations and the $\widehat{a}_k$ in green. Similarly, Fig. \ref{fig:mean_expansion} also shows values of $|A_k|/k$ (the graph expansion) also taken over 500 realizations.

\begin{figure}[h]\vspace{-3mm}
\centering
\includegraphics[width=0.5\textwidth]{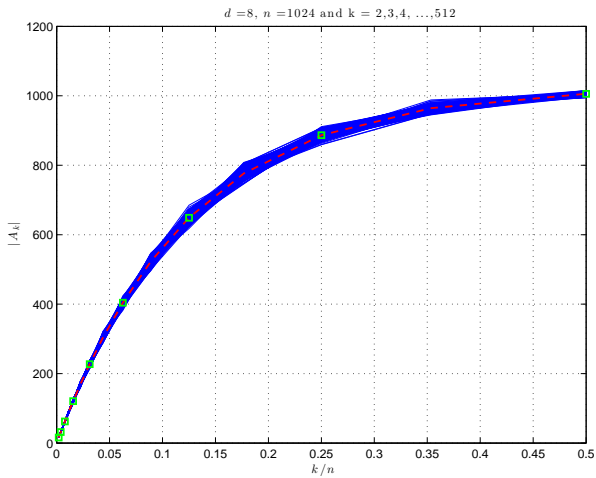}
\caption{For fixed $d=8$ and $n=2^{10}$, over $500$ realizations we plot (in blue) the cardinalities of the index sets of nonzeros in a given number of set sizes, $k$. The dotted red curve is mean of the simulations and the green squares are the $\hat{a}_k$.}
\label{fig:mean_neighbour}
\end{figure}

\begin{figure}[h]\vspace{-3mm}
\centering
\includegraphics[width=0.5\textwidth]{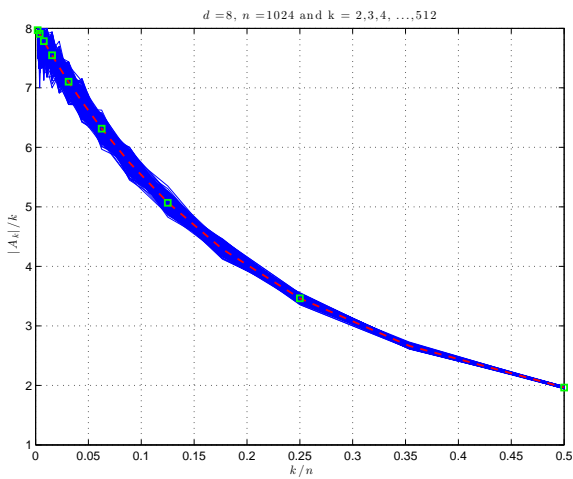}
\caption{For fixed $d=8$ and $n=2^{10}$, over $500$ realizations we plot (in blue) the graph expansion for a given input set size $k$. The dotted red curve is mean of the simulations and the green squares are the $\hat{a}_k/k$.}
\label{fig:mean_expansion}
\end{figure}

Theorem \ref{thm:prob_bound_1set_expansion} also claims that the $\hat{a}_s$ are the expected values of the cardinalities of the union of $s$ sets. We give a brief sketch of its proof in Section \ref{sec:tools} in terms of the maximum likelihood and  empirical illustrate the accuracy of the result in Fig. \ref{fig:compare} where we show the relative error between $\hat{a}_k$ and the mean values of the $a_k, \bar{a}_k$, realized over $500$ runs, to be less than $10^{-3}$.

\begin{figure}[h]\vspace{-3mm}
\centering
\includegraphics[width=0.5\textwidth]{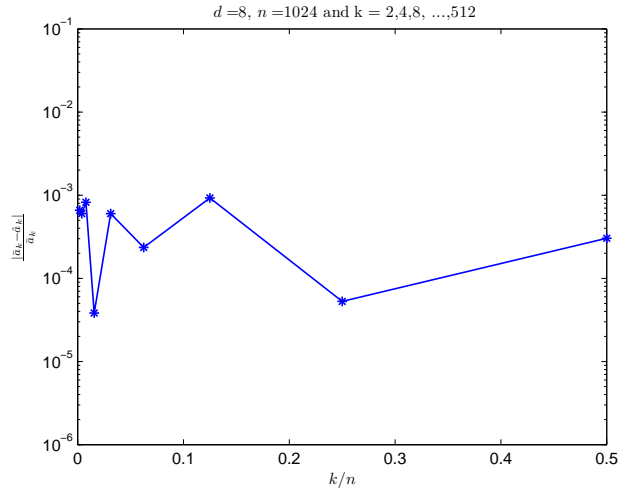}
\caption{For fixed $d=8$ and $n=2^{10}$, over $500$ realizations the relative error between the mean values of $a_k$ (referred to as $\bar{a}_k$) and the $\hat{a}_k$ from Equation \eqref{eq:aj_system_unconstrained} of Theorem \ref{thm:prob_bound_1set_expansion}.}
\label{fig:compare}
\end{figure}

\begin{figure}[h]\vspace{-3mm}
\centering
\includegraphics[width=0.5\textwidth]{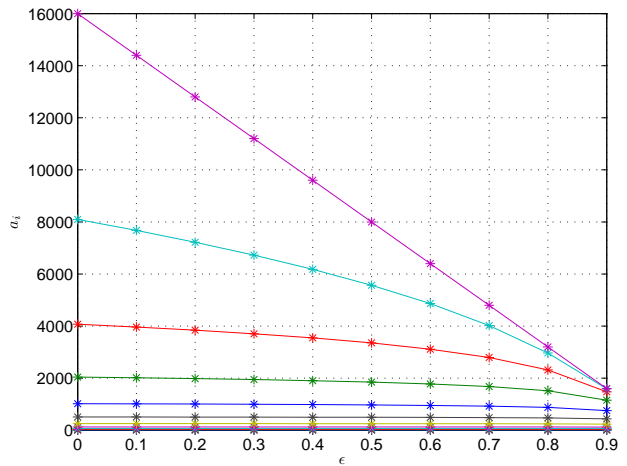}
\caption{Values of $a_i$as a function of $\epsilon\in [0,1)$ for $a_k:=(1-\epsilon)\hat{a}_k$  with $d=8$, $k=2\times 10^3$ and $n=2^{20}$. For this choice of $d,k,n$ there are twelve levels of dyadic splits resulting in $a_i$ for $i=2^j$ for $j = 0,\ldots,\lceil\log_2k \rceil=12$.  The highest curve corresponds to $a_i$ for $i=2^{12}$, the next highest curve corresponds to $i=2^{11}$, and continuing in decreasing magnitude with decreasing subscript values.}
\label{fig:plot_a_vs_eps}
\end{figure}

Fig. \ref{fig:plot_a_vs_eps} shows representative values of $a_i$ from \eqref{eq:aj_system_constrained} for $a_k:=(1-\epsilon)\hat{a}_k$ as a function of $\epsilon$ for $d=8$, $k=2\times 10^3$, and $n=2^{20}$.  Each of the $a_i$ decrease
smoothly towards $d$, but with $a_i$ for smaller values if $i$ varying less than for larger values of $i$.

For fixed $0<\e<1/2$ and for small but fixed $d$, $\r^{exp}(\d;d,\e)$ in Corollary \ref{cor:prob_expander_existence2} is a function of $\d$ for each $d$ and $\e$, is a phase transition function in the $(\d,\r)$ plane. Below the curve of $\r^{exp}(\d;d,\e)$ the probability in \eqref{eq:prob_expander_existence2} goes to one exponentially in $n$ as the problem size grows. That is if $A$ is drawn at random with $d ~1$s or $d ~\pm 1$s in each column and having parameters $(k,n,N)$ that fall below the curve of $\r^{exp}(\d;d,\e)$ then we say it is from the class of matrices in Definition \ref{def:1set_expansion} with probability approaching one exponentially in $n$. In terms of $|\Gamma(X)|$ for $X\subset U$ and $|X|\leq k$, Corollary \ref{cor:prob_expander_existence2} say that the probability $|\Gamma(X)| \geq (1-\e) dk$ goes to one exponentially in $n$ if the parameters of our graph lies in the region below $\r^{exp}(\d;d,\e)$. This implies that if we draw a random bipartite graphs that has parameters in the region below the curve of $\r^{exp}(\d;d,\e)$ then with probability approaching one exponentially in $n$ that graph is a lossless $(k,d,\e)$-expander.

\begin{figure}[h]\vspace{-3mm}
\centering
\includegraphics[width=0.5\textwidth]{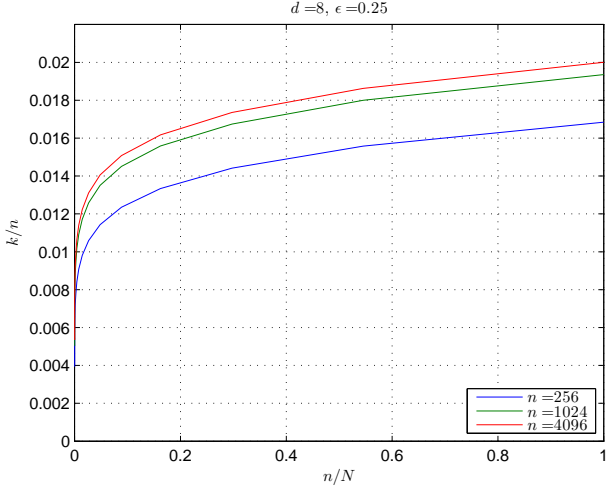}
\caption{Phase transition plots of $\r^{exp}(\d;d,\e)$ for fixed $d=8$ and $\e = 1/4$ with $n$ varied.}
\label{fig:phase_transition_n}
\end{figure}

Fig. \ref{fig:phase_transition_n} shows a plot of what $\r^{exp}(\d;d,\e)$ converge to for different values of $n$ with $\e$ and $d$ fixed; Fig. \ref{fig:phase_transition_d} shows a plot of what $\r^{exp}(\d;d,\e)$ converge to for different values of $d$ with $\e$ and $n$ fixed; while Fig. \ref{fig:phase_transition_e} shows plots of what $\r^{exp}(\d;d,\e)$ converge to for different values of $\e$ with $n$ and $d$ fixed. It is interesting to note how increasing $d$ increases the phase transition up to a point then it decreases the phase transition. Essentially beyond $d=16$ there is no gain in increasing $d$. This vindicates the use of small $d$ in most of the numerical simulations involving the class of matrices considered here. Note the vanishing sparsity as the problem size $(k,n,N)$ grows while $d$ is fixed to a small value of $8$. In their GPU implementation \cite{blanchard2012gpu} Blanchard and Tanner observed that SSE with $d=7$ has a phase transition for numerous sparse approximation algorithms that is consistent with dense Gaussian matrices, but with dramatically faster implementation.

\begin{figure}[h]\vspace{-3mm}
\centering
\includegraphics[width=0.5\textwidth]{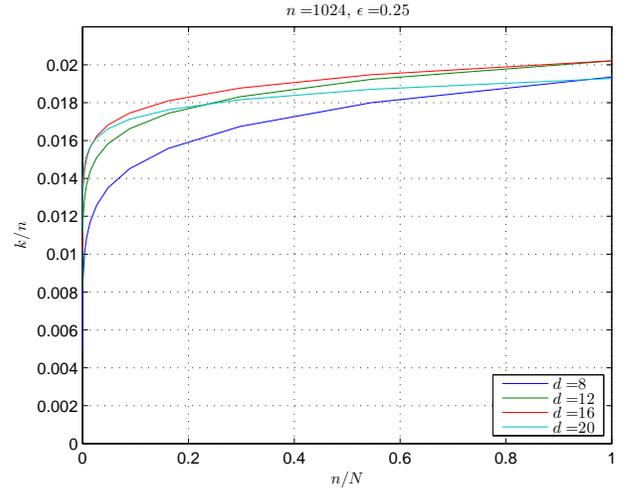}
\caption{Phase transition plots of $\r^{exp}(\d;d,\e)$ for fixed $\e=1/6$ and $n = 2^{10}$ with $d$ varied.}
\label{fig:phase_transition_d}
\end{figure}

As afore-stated Corollary \ref{cor:prob_expander_existence2} follows from Theorem \ref{thm:prob_bound_1set_expansion}, alternatively Corollary \ref{cor:prob_expander_existence2} can be arrived at based on probabilistic constructions of expander graphs given by Proposition \ref{pro:para_opt} below. This proposition and its proof can be traced back to Pinsker in \cite{pinsker1973complexity} but for more recent proofs see \cite{berinde2009advances,capalbo2002randomness}.
\begin{proposition}
\label{pro:para_opt}
 For any $N/2 \geq k \geq 1, ~\epsilon > 0$ there exists a lossless $(k,d,\epsilon)$-expander with $$\displaystyle d = \mathcal{O}\left(\log\left(N/k\right)/\epsilon\right) \quad \mbox{and} \quad n = \mathcal{O}\left(k\log\left(N/k\right)/\epsilon^2\right).$$
\end{proposition}

\begin{figure}[h]\vspace{-3mm}
\centering
\includegraphics[width=0.5\textwidth]{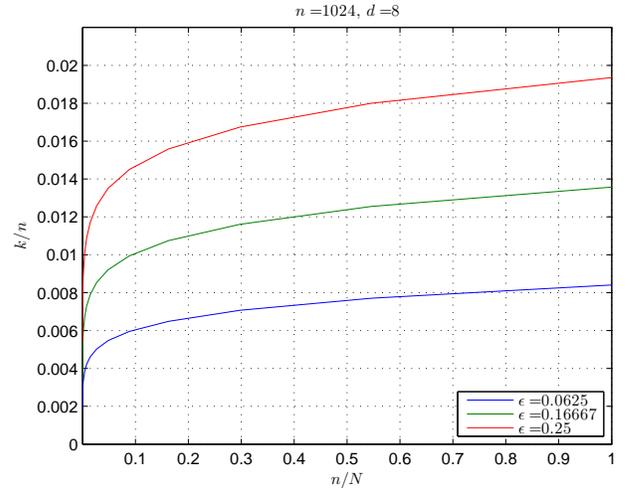}
\caption{Phase transition plots of $\r^{exp}(\d;d,\e)$ for fixed $d=8$ and $n = 2^{10}$ with $\e$ varied.}
\label{fig:phase_transition_e}
\end{figure}

To put our results in perspective, we compare them to the alternative construction in \cite{berinde2009advances} which led to Corollary \ref{cor:prob_expander_existenceBI},  whose proof is given in Section \ref{sec:prf:expander_existence} of the Appendix. Fig. \ref{fig:phaseT} compares the phase transitions resulting from our construction to that presented in \cite{berinde2009advances}, but we must point out however, that the proof in \cite{berinde2009advances} was not aimed for a tight bound.

\begin{figure}[h]\vspace{-3mm}
\centering
\includegraphics[width=0.5\textwidth]{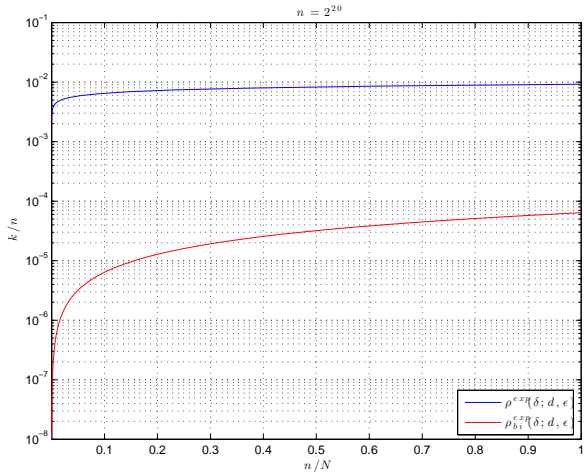}
\caption{A comparison of $\r^{exp}$ in Theorem \ref{thm:prob_bound_1set_expansion} to $\r_{bi}^{exp}$ of Corollary \ref{cor:prob_expander_existenceBI} derived using the construction based on Corollary \ref{cor:prob_expander_existenceBI}.}
\label{fig:phaseT}
\end{figure}

\begin{corollary}\vspace{-2mm}
\label{cor:prob_expander_existenceBI}
Consider a bipartite graph $G=(U,V,E)$ with left vertices $|U|=N$, right vertices $|V|=n$ and left degree $d$. Fix $0<\e<1/2$ and $d$, as $(k,n,N) \rightarrow \infty$ while $k/n \rightarrow \r \in (0,1)$ and $n/N \rightarrow \d \in (0,1)$ then for $\r < (1-\g)\r_{bi}^{exp}(\d;d,\e)$ and $\g>0$
\begin{equation}
\label{eq:prob_expander_existenceBI}
\hbox{Prob}\left(G ~\mbox{fails to be an expander}\right) \rightarrow 0
\end{equation}
exponentially in $n$, where $\r_{bi}^{exp}(\d;d,\e)$ is the largest limiting value of $k/n$ for which
\begin{equation}
 \label{eq:rho_expBI}
    \Psi\left(k,n,N;d,\e\right) = 0
\end{equation}
with $\displaystyle \Psi\left(k,n,N;d,\e\right) = \mbox{H}\left(\frac{k}{N}\right) + \frac{dk}{N}\mbox{H}\left(\e\right) + \frac{\e dk}{N} \log\left(\frac{dk}{n}\right)$.
\end{corollary}

\subsection{Key Lemmas}\label{sec:tools}
The following set of lemmas, propositions and corollaries form the building blocks of the proof of our main results to be presented in Section \ref{sec:proof}.

For one fixed set of columns of $A$, denoted $A_S$,  the probability in \eqref{eq:prob_bound_1set_expansion} can be understood as the cardinality of the unions of nonzeros in the columns of $A_S$.  Our analysis of this probability follows from a nested unions of subsets using a {\em dyadic splitting} technique. Given a starting set of columns we recursively split the number of columns from this set, and the resulting sets, into two sets composed of the ceiling and floor of half of the number of columns of the set we split. In other words, given a starting support set $S$ (referred to as the parent set), we split it into two disjoint sets of size $\lceil s/2 \rceil$ and $\lfloor s/2 \rfloor$ (referred to as children). Then we union the nonzero elements in the columns indexed by the children sets to get $A_{\lceil s/2 \rceil}$ and $A_{\lfloor s/2 \rfloor}$. We continue this process until at a level when the cardinalities in each child set is at most two. Resulting from this type of splitting is a regular binary tree where the size of each child is either the ceiling or the floor of the size of it's parent set. The root of the binary is our starting set $S$ or $A_s$ which we refer to as {\em level 0}. Then, the splitting, as described above proceeds till {\em level} $\lceil \log_2\rceil - 1$, see Fig. \ref{fig:splitting}.

\begin{figure}[h]
\centering
\includegraphics[width=0.5\textwidth]{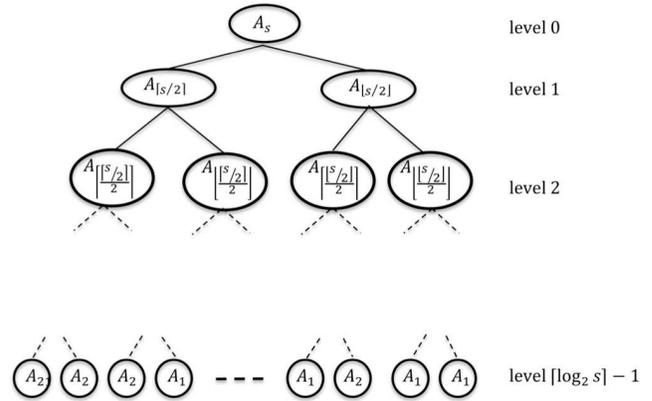}
\caption{The binary splitting of the support $S$ that indexes the columns of $A_S$ resulting into a regular binary tree. Here we show the set of neighbors, $A_s$, and the resulting sets from the splitting. Each child has either the ceiling or the floor of it's parent's number of columns. The leaves of the tree, which are at {\em level} $\lceil\log_2s\rceil - 1$ of the tree, have sets that are composed of union of nonzero elements in at most two columns, $A_2$.}
\label{fig:splitting}
\end{figure}

The probability of interest becomes a product of the probabilities involving all the children from the dyadic splitting of the original set $S$, the index of the union of the nonzero elements forming $A_s$.

The computation of the probability in \eqref{eq:prob_bound_1set_expansion} involves the computation of the probability of the cardinality of the intersection of two sets. This probability is given by Lemma \ref{lem:intersect_prob} and Corollary \ref{cor:intersect_prob} below.

\begin{lemma}
\label{lem:intersect_prob}
Let $B,~B_1,~B_2 \subset [n]$ where $\left|B_1\right|=b_1$, $\left|B_2\right|=b_2$, $B = B_1\cup B_2$ and $|B|=b$. Also let $B_1$ and $B_2$ be drawn uniformly at random, independent of each other, and define $\hbox{P}_n\left(b,b_1,b_2\right) := \hbox{Prob}\left(\left|B_1\cap B_2\right| = b_1+b_2-b\right)$, then 
\begin{equation}
\label{eq:intersect_prob}
\hbox{P}_n\left(b,b_1,b_2\right) = \binom{b_1}{b_1+b_2-b} \binom{n-b_1}{b-b_1} \binom{n}{b_2}^{-1}.
\end{equation}
\end{lemma}

\begin{proof}
Given $B_1,B_2 \subset [n]$ where $\left|B_1\right|=b_1$ and $\left|B_2\right|=b_2$ are drawn uniformly at random, independent of each other, we calculate $\hbox{Prob}\left(\left|B_1\cap B_2\right| = z\right)$ where $z=b_1+b_2-b$. Without loss of generality consider drawing $B_1$ first, then the probability that the draw of $B_2$ intersecting $B_1$ will have cardinality $z$, i.e. $\hbox{Prob}\left(\left|B_1\cap B_2\right| = z\right)$, is the size of the event of drawing $B_2$ intersecting $B_1$ by $z$ divided by the size of the sample space of drawing $B_2$ from $[n]$, which are given by $\binom{b_1}{z} \cdot \binom{n-b_1}{b_2-z}$ and $\binom{n}{b_2}$ respectively. Rewriting the division as a product with the divisor raised to a negative power and replacing $z$ by $b_1+b_2-b$ gives \eqref{eq:intersect_prob}.
\end{proof}

\begin{corollary}
\label{cor:intersect_prob}
If two sets, $B_1,B_2 \subset [n]$ are drawn uniformly at random, independent of each other, and $B = B_1\cup B_2$
\begin{multline}
\label{eq:intersect_prob2}
 \hbox{Prob}\left(|B| = b\right) = \hbox{P}_n\left(b,b_1,b_2\right) \times \\ \hbox{Prob}\left(|B_1| = b_1\right) \cdot \hbox{Prob}\left(|B_2|=b_2\right)
\end{multline}
\end{corollary}

\begin{proof}
$\hbox{Prob}\left(|B| = b\right) = \hbox{Prob}\left(|B_1\cup B_2| = b\right)$ by definition. As a consequence of the inclusion-exclusion principle
\begin{multline}
\label{eq:intersect_prob3}
\hbox{Prob}\left(|B_1\cup B_2| = b\right) = \hbox{Prob}\left(|B_1\cap B_2| = b_1+b_2-b\right) \\ \times \hbox{Prob}\left(|B_1| = b_1\right) \cdot \hbox{Prob}\left(|B_2|=b_2\right).
\end{multline}
We use Lemma \ref{lem:intersect_prob} to replace $\hbox{Prob}\left(|B_1\cap B_2| = b_1+b_2-b\right)$ in \eqref{eq:intersect_prob3} by $\hbox{P}_n\left(b,b_1,b_2\right)$ leading to the required result.
\end{proof}

In the binary tree resulting from our dyadic splitting scheme the number of columns in the two children of a parent node is the ceiling and the floor of half of the number of columns of the parent node. At each level of the split the number of columns of the children of that level differ by one. The enumeration of these two quantities at each level of the splitting process is necessary in the computation of the probability of \eqref{eq:prob_bound_1set_expansion}. We state and prove what we refer to a {\em dyadic splitting lemma}, Lemma \ref{lem:num_size_split}, which we later use to enumerate these two quantities - the sizes (number of columns) of the children and the number of children with a given size at each level of the split.
\begin{lemma}
\label{lem:num_size_split}
Let $S$ be an index set of cardinality $s$.  For any level $j$ of the dyadic splitting, $j = 0,\ldots,\lceil \log_2 s \rceil -1$, the set $S$ is decomposed into disjoint sets each having cardinality $Q_j = \big{\lceil}\frac{s}{2^j}\big{\rceil}$ or $R_j = Q_j - 1$. Let $q_j$ sets have cardinality $Q_j$ and $r_j$ sets have cardinality $R_j$, then
\begin{align}
\label{eq:split_num_size}
q_j = s - 2^j\cdot\Big{\lceil}\frac{s}{2^j}\Big{\rceil} + 2^j, \quad \mbox{and} \quad r_j = 2^j - q_j.
\end{align}
\end{lemma}

\begin{proof}
At every node on the binary tree the children have either of two sizes (number of columns) of the floor and ceiling of half the sizes of their parents and these sizes differ at most by 1, that is at level $j$ of the splitting we have at most 2 different sizes. We define these sizes, $Q_j$ and $R_j$, in terms of two arbitrary integers, $m_1$ and $m_2$, as follows.
\begin{equation}
\label{eq:split_num_cols}
Q_j = \frac{s}{2^j} + \frac{m_1}{2^j} \quad \mbox{and} \quad R_j = \frac{s}{2^j} + \frac{m_2}{2^j}.
\end{equation}
Because of the nature of our splitting scheme we have $R_j = Q_j - 1$ which implies that $m_1$ and $m_2$ must satisfy the relation
\begin{equation}
\label{eq:split_num_cols1}
\frac{m_1-m_2}{2^j} = 1.
\end{equation}
Now let $q_j$ and $r_j$ be the number of children with $Q_j$ and $R_j$ number of columns respectively. Therefore,
\begin{equation}
\label{eq:split_num_cols2}
q_j + r_j = 2^j.
\end{equation}
At each level $j$ of the splitting the following condition must be satisfied
\begin{equation}
\label{eq:split_num_cols3}
q_j \cdot Q_j + r_j \cdot R_j = s.
\end{equation}
To find $m_1, ~m_2, ~q_j$ and $r_j$, from \eqref{eq:split_num_cols} we substitute for $Q_j$ and $R_j$ in \eqref{eq:split_num_cols3} to have
\begin{align}
\label{eq:split_num_cols4a}
q_j \cdot \left(\frac{s}{2^j} + \frac{m_1}{2^j}\right) + r_j \cdot \left(\frac{s}{2^j} + \frac{m_2}{2^j}\right) & = s,\\
\label{eq:split_num_cols4b}
2^{-j} q_j s + 2^{-j} q_j m_1 + 2^{-j} r_j s + 2^{-j} r_j m_2 & = s,\\
\label{eq:split_num_cols4c}
2^{-j} \left(q_j + r_j\right) s + 2^{-j} \left(q_j m_1 + r_j m_2\right) & = s,\\
\label{eq:split_num_cols4d}
s + 2^{-j} \left(q_j m_1 + r_j m_2\right) & = s,\\
\label{eq:split_num_cols4e}
q_j m_1 + r_j m_2 & = 0.
\end{align}
We expanded the brackets from \eqref{eq:split_num_cols4a} to \eqref{eq:split_num_cols4b} and simplified from \eqref{eq:split_num_cols4b} to \eqref{eq:split_num_cols4c}. We simplify the first term of \eqref{eq:split_num_cols4c} using \eqref{eq:split_num_cols2} to get \eqref{eq:split_num_cols4d} and we simplified this to get \eqref{eq:split_num_cols4e}.

Equation \eqref{eq:split_num_cols1} yields
\begin{equation}
\label{eq:split_num_cols5}
m_1 = m_2 + 2^j.
\end{equation}
Substituting this in \eqref{eq:split_num_cols4e} yields
\begin{align}
\label{eq:split_num_cols6a}
q_j \left(m_2 + 2^j\right) + r_j m_2 & = 0,\\
\label{eq:split_num_cols6b}
\left(q_j + r_j\right) m_2 + 2^j q_j & = 0,\\
\label{eq:split_num_cols6c}
2^j\left(q_j + m_2\right) & = 0.
\end{align}
From \eqref{eq:split_num_cols6a} to \eqref{eq:split_num_cols6b} we expanded the brackets and rearranged the terms and used \eqref{eq:split_num_cols2} to simplify to \eqref{eq:split_num_cols6c}. Using \eqref{eq:split_num_cols6c} and \eqref{eq:split_num_cols5} respectively we have
\begin{equation}
\label{eq:split_num_cols7}
m_2 = -q_j \quad \mbox{and} \quad m_1 = 2^j - q_j = r_j.
\end{equation}
Substituting this in \eqref{eq:split_num_cols} we have
\begin{equation}
\label{eq:split_num_cols8}
Q_j = \frac{s-q_j}{2^j} + 1 \quad \mbox{and} \quad R_j = \frac{s-q_j}{2^j}.
\end{equation}
Equating this value of $Q_j$ to its defined value in the statement of the lemma gives
\begin{equation}
\label{eq:split_num_cols9}
\frac{s-q_j}{2^j} + 1 = \Big{\lceil}\frac{s}{2^j}\Big{\rceil} \quad \Rightarrow \quad q_j = s - 2^j\cdot\Big{\lceil}\frac{s}{2^j}\Big{\rceil} + 2^j.
\end{equation}
Therefore, from \eqref{eq:split_num_cols7} we use \eqref{eq:split_num_cols9} to have
\begin{equation}
\label{eq:split_num_cols10}
r_j = 2^j - q_j \quad \Rightarrow \quad r_j = 2^j \cdot \Big{\lceil}\frac{s}{2^j}\Big{\rceil} - s,
\end{equation}
which concludes the proof.
\end{proof}

The bound in \eqref{eq:prob_bound_1set_expansion} is derived using a large deviation analysis of the nested probabilities which follow from the dyadic splitting in Corollary \ref{cor:intersect_prob}.   The large deviation analysis of \eqref{eq:intersect_prob} at each stage involves its large deviation exponent $\psi_n(\cdot)$, which follows from
Stirling's inequality bounds on the combinatorial product of \eqref{eq:intersect_prob}. Lemma \ref{lem:psi_n_behaviour}  establishes a few properties of $\psi_n(\cdot)$ while Lemma \ref{lem:Psi_n_bound} shows how the various $\psi_n(\cdot)$'s at a given dyadic splitting level can be combined into a relatively simple expression.

\begin{lemma}
\label{lem:psi_n_behaviour}
Define
\begin{multline}
\label{eq:psi_n_def}
\psi_n(x,y,z) :=y\cdot\hbox{H}\left(\frac{x-z}{y}\right) + (n-y)\cdot\hbox{H}\left(\frac{x-y}{n-y}\right) \\
- n\cdot\hbox{H}\left(\frac{z}{n}\right)
,
\end{multline}
then for $n>x>y$ we have that
\begin{align}
\label{eq:psi_n_property1}
& \mbox{for } y>z \quad \psi_n(x,y,y) \leq \psi_n(x,y,z) \leq \psi_n(x,z,z); \\
\label{eq:psi_n_property2}
& \mbox{for } x>z \quad \psi_n(x,y,y) > \psi_n(z,y,y); \\
 \label{eq:psi_n_property3}
 & \mbox{for } 1/2< \a \leq 1 \quad \psi_n(x,y,y) < \psi_n(\a x,\a y,\a y).
\end{align}
\end{lemma}

\begin{proof}
We start with Property \eqref{eq:psi_n_property1} and first show that the left inequality holds. If we substitute $y$ for $z$ in  \eqref{eq:psi_n_def} with $y>z$ we reduce the first and last terms of \eqref{eq:psi_n_def} while we increase the middle term of \eqref{eq:psi_n_def} which makes $\psi_n(x,y,y) \leq \psi_n(x,y,z)$. For second inequality we replace $y$ by $z$ in \eqref{eq:psi_n_def} with $y>z$ we increase the first and the last terms of \eqref{eq:psi_n_def} and reduce the middle term which makes $\psi_n(x,y,z) \leq \psi_n(x,z,z)$. This concludes the proof for \eqref{eq:psi_n_property1}.

Property \eqref{eq:psi_n_property2} states that for fixed $y$, $\psi_n(x,y,y)$ is monotonically increasing in its first argument. To prove \eqref{eq:psi_n_property2} we use the condition $n>x>y$ to ensure that $\hbox{H}(p)$ increases monotonically with $p$, which implies that the first and last terms of \eqref{eq:psi_n_def} increase with $x$ for fixed $y$ while the second term remains constant.

Property \eqref{eq:psi_n_property3} means that $\psi_n(x,y,y)$ is monotonically decreasing in $x$ and $y$. For the proof we show that for $1/2< \a \leq 1$ the difference $\psi_n(\a x,\a y,\a y) - \psi_n(x,y,y) > 0$. Using \eqref{eq:psi_n_def} we write out clearly what the difference, $\psi_n(\a x,\a y,\a y) - \psi_n(x,y,y)$, is as follows.

 \begin{align}
 \label{eq:psi_n_difference1}
 &\a y\hbox{H}\left(\frac{\a x-\a y}{\a y}\right) + (n-\a y)\hbox{H}\left(\frac{\a x-\a y}{n-\a y}\right) - n\hbox{H}\left(\frac{\a y}{n}\right)\nonumber\\
 & - y\hbox{H}\left(\frac{x-y}{y}\right) - (n-y)\hbox{H}\left(\frac{x-y}{n-y}\right) + n\hbox{H}\left(\frac{y}{n}\right)\\
 \label{eq:psi_n_difference2}
 & = \a y\hbox{H}\left(\frac{x-y}{y}\right) + n\hbox{H}\left(\frac{\a x-\a y}{n-\a y}\right) - \a y  \hbox{H}\left(\frac{\a x-\a y}{n-\a y}\right)\nonumber\\
 &\quad - n\hbox{H}\left(\frac{\a y}{n}\right) - y\hbox{H}\left(\frac{x-y}{y}\right) - n\hbox{H}\left(\frac{x-y}{n-y}\right) \nonumber\\
 & \quad + y\hbox{H}\left(\frac{x-y}{n-y}\right) + n\hbox{H}\left(\frac{y}{n}\right)\\
 \label{eq:psi_n_difference3}
 & = \a y\hbox{H}\left(\frac{x-y}{y}\right) - \a y\hbox{H}\left(\frac{\a x-\a y}{n-\a y}\right) - y\hbox{H}\left(\frac{x-y}{y}\right)\nonumber\\
 &\quad + y\hbox{H}\left(\frac{x-y}{n-y}\right) + n\hbox{H}\left(\frac{y}{n}\right) - n\hbox{H}\left(\frac{\a y}{n}\right) \nonumber\\
 & \quad + n\hbox{H}\left(\frac{\a x-\a y}{n-\a y}\right) - n\hbox{H}\left(\frac{x-y}{n-y}\right)
 \end{align}
 From \eqref{eq:psi_n_difference1} to \eqref{eq:psi_n_difference2} we expanded brackets and simplified, while from \eqref{eq:psi_n_difference2} to \eqref{eq:psi_n_difference3} we rearranged the terms for easy comparison.

 Again $n>x>y$ ensures that the arguments of $\hbox{H}(\cdot)$ are strictly less than half and $\hbox{H}(p)$ increases monotonically with $p$. In \eqref{eq:psi_n_difference3} the difference of the first two terms in the first row is positive while the difference of the second two terms is negative. However, the whole sum of the first four terms is negative but very close to zero when $\a$ is close to one which is the regime that we will be considering. The difference of the last two terms in the second row is positive while the difference of the terms on bottom row is negative but due to the concavity and steepness of the Shannon entropy function the first positive difference is larger hence the sum of last four terms is positive. Since we can write $n=cy$ with $c>1$ being an arbitrarily constant, then the positive sum in the second four terms dominates the negative sum in the first four terms. This gives the required results and hence concludes this proof and the proof of Lemma \ref{lem:psi_n_behaviour}.
\end{proof}

\begin{lemma}
 \label{lem:Psi_n_bound}
Given $\psi_n(\cdot)$ as defined in \eqref{eq:psi_n_def} then the following bound holds.
\begin{align}
    \label{eq:Psi_n_bound}
    & \sum_{j=0}^{\lceil\log_2(s)\rceil-2} \left[ q_j \cdot \psi_n\left(a_{Q_j},a_{\big{\lceil}\frac{Q_j}{2}\big{\rceil}},a_{\big{\lfloor}\frac{Q_j}{2}\big{\rfloor}}\right)\right. + \nonumber \\
    & \left. r_j \cdot \psi_n\left(a_{R_j},a_{\big{\lceil}\frac{R_j}{2}\big{\rceil}},a_{\big{\lfloor}\frac{R_j}{2}\big{\rfloor}}\right) \right]
    + q_{\lceil\log_2(s)\rceil-1} \cdot \psi_n\left(a_{2},d,d\right) \nonumber \\
    & \quad \leq \sum_{j=0}^{\lceil\log_2(s)\rceil-1} 2^j \cdot \psi_n\left(a_{Q_j},a_{\big{\lfloor}\frac{R_j}{2}\big{\rfloor}},a_{\big{\lfloor}\frac{R_j}{2}\big{\rfloor}}\right),
\end{align}
where $a_{\frac{R_{\lceil\log_2(s)\rceil-1}}{2}} = d$.
\end{lemma}

\begin{proof}
The quantity inside the left hand side summation in \eqref{eq:Psi_n_bound}, i.e.
\begin{multline}
\label{eq:prob_1set_ssparse_exp4a}
q_j \cdot \psi_n\left(a_{Q_j},a_{\big{\lceil}\frac{Q_j}{2}\big{\rceil}},a_{\big{\lfloor}\frac{Q_j}{2}\big{\rfloor}}\right) \\ + r_j \cdot \psi_n\left(a_{R_j},a_{\big{\lceil}\frac{R_j}{2}\big{\rceil}},a_{\big{\lfloor}\frac{R_j}{2}\big{\rfloor}}\right),
\end{multline}
is equal to the following if we replace $q_j$ and $r_j$ by their values given in Lemma \ref{lem:num_size_split}.
\begin{align}
    \label{eq:prob_1set_ssparse_exp4b}
    & \left( s - 2^j\Big{\lceil}\frac{s}{2^j}\Big{\rceil} + 2^j\right) \cdot \psi_n\left(a_{Q_j},a_{\big{\lceil}\frac{Q_j}{2}\big{\rceil}},a_{\big{\lfloor}\frac{Q_j}{2}\big{\rfloor}}\right) \\
    & + \left( 2^j\Big{\lceil}\frac{s}{2^j}\Big{\rceil} - s \right) \cdot \psi_n\left(a_{R_j},a_{\big{\lceil}\frac{R_j}{2}\big{\rceil}},a_{\big{\lfloor}\frac{R_j}{2}\big{\rfloor}}\right)\nonumber \\
    \label{eq:prob_1set_ssparse_exp4c}
    & < \left( s - 2^j\Big{\lceil}\frac{s}{2^j}\Big{\rceil} + 2^j\right) \cdot \psi_n\left(a_{Q_j},a_{\big{\lfloor}\frac{Q_j}{2}\big{\rfloor}},a_{\big{\lfloor}\frac{Q_j}{2}\big{\rfloor}}\right) \\
    & + \left( 2^j\Big{\lceil}\frac{s}{2^j}\Big{\rceil} - s \right) \cdot \psi_n\left(a_{R_j},a_{\big{\lfloor}\frac{R_j}{2}\big{\rfloor}},a_{\big{\lfloor}\frac{R_j}{2}\big{\rfloor}}\right).\nonumber\\
    \label{eq:prob_1set_ssparse_exp4d}
    & < \left( s - 2^j\Big{\lceil}\frac{s}{2^j}\Big{\rceil} + 2^j\right) \cdot \psi_n\left(a_{Q_j},a_{\big{\lfloor}\frac{R_j}{2}\big{\rfloor}},a_{\big{\lfloor}\frac{R_j}{2}\big{\rfloor}}\right) \\
    & + \left( 2^j\Big{\lceil}\frac{s}{2^j}\Big{\rceil} - s \right) \cdot \psi_n\left(a_{R_j},a_{\big{\lfloor}\frac{R_j}{2}\big{\rfloor}},a_{\big{\lfloor}\frac{R_j}{2}\big{\rfloor}}\right).\nonumber\\
    \label{eq:prob_1set_ssparse_exp4e}
    & < \left( s - 2^j\Big{\lceil}\frac{s}{2^j}\Big{\rceil} + 2^j\right) \cdot \psi_n\left(a_{Q_j},a_{\big{\lfloor}\frac{R_j}{2}\big{\rfloor}},a_{\big{\lfloor}\frac{R_j}{2}\big{\rfloor}}\right) \\
    & + \left( 2^j\Big{\lceil}\frac{s}{2^j}\Big{\rceil} - s \right) \cdot \psi_n\left(a_{Q_j},a_{\big{\lfloor}\frac{R_j}{2}\big{\rfloor}},a_{\big{\lfloor}\frac{R_j}{2}\big{\rfloor}}\right).\nonumber\\
    \label{eq:prob_1set_ssparse_exp4f}
    & = 2^j \cdot \psi_n\left(a_{Q_j},a_{\big{\lfloor}\frac{R_j}{2}\big{\rfloor}},a_{\big{\lfloor}\frac{R_j}{2}\big{\rfloor}}\right).
\end{align}
From \eqref{eq:prob_1set_ssparse_exp4b} to \eqref{eq:prob_1set_ssparse_exp4c} we upper bounded $\psi_n\left(a_{Q_j},a_{\big{\lceil}\frac{Q_j}{2}\big{\rceil}},a_{\big{\lfloor}\frac{Q_j}{2}\big{\rfloor}}\right)$ by $\psi_n\left(a_{Q_j},a_{\big{\lfloor}\frac{Q_j}{2}\big{\rfloor}},a_{\big{\lfloor}\frac{Q_j}{2}\big{\rfloor}}\right)$ and $\psi_n\left(a_{R_j},a_{\big{\lceil}\frac{R_j}{2}\big{\rceil}},a_{\big{\lfloor}\frac{R_j}{2}\big{\rfloor}}\right)$ by $\psi_n\left(a_{R_j},a_{\big{\lfloor}\frac{R_j}{2}\big{\rfloor}},a_{\big{\lfloor}\frac{R_j}{2}\big{\rfloor}}\right)$ using \eqref{eq:psi_n_property1} of Lemma \ref{lem:psi_n_behaviour}. We then upper bounded $\psi_n\left(a_{Q_j},a_{\big{\lfloor}\frac{Q_j}{2}\big{\rfloor}},a_{\big{\lfloor}\frac{Q_j}{2}\big{\rfloor}}\right)$ by $\psi_n\left(a_{Q_j},a_{\big{\lfloor}\frac{R_j}{2}\big{\rfloor}},a_{\big{\lfloor}\frac{R_j}{2}\big{\rfloor}}\right)$, from \eqref{eq:prob_1set_ssparse_exp4c} to \eqref{eq:prob_1set_ssparse_exp4d}, again using \eqref{eq:psi_n_property1} of Lemma \ref{lem:psi_n_behaviour}. From \eqref{eq:prob_1set_ssparse_exp4d} to \eqref{eq:prob_1set_ssparse_exp4e}, using \eqref{eq:psi_n_property2} of Lemma \ref{lem:psi_n_behaviour}, we bounded $\psi_n\left(a_{R_j},a_{\big{\lfloor}\frac{R_j}{2}\big{\rfloor}},a_{\big{\lfloor}\frac{R_j}{2}\big{\rfloor}}\right)$ by $\psi_n\left(a_{Q_j},a_{\big{\lfloor}\frac{R_j}{2}\big{\rfloor}},a_{\big{\lfloor}\frac{R_j}{2}\big{\rfloor}}\right)$. For the final step from \eqref{eq:prob_1set_ssparse_exp4e} to \eqref{eq:prob_1set_ssparse_exp4f} we factored out $\psi_n\left(a_{Q_j},a_{\big{\lfloor}\frac{R_j}{2}\big{\rfloor}},a_{\big{\lfloor}\frac{R_j}{2}\big{\rfloor}}\right)$ and then simplified.

Using $q_{\lceil\log_2(s)\rceil-1} + r_{\lceil\log_2(s)\rceil-1} = 2^{\lceil\log_2(s)\rceil-1}$ we bound $q_{\lceil\log_2(s)\rceil-1}$ by $2^{\lceil\log_2(s)\rceil-1}$. Then we add this to the summation of \eqref{eq:prob_1set_ssparse_exp4e} for $j=0,\ldots,\lceil\log_2(s)\rceil-2$ establishing the bound of Lemma \ref{lem:Psi_n_bound}.
\end{proof}

Now we state and prove a lemma about the quantities $a_i$. During the proof we will make a statement about the $a_i$ using their expected values $\hat{a}_i$ which follows from a maximum likelihood analogy.
\begin{lemma}
\label{lem:max_Psi_n_bound}
The problem
\begin{equation}
\label{eq:prob_1set_ssparse_expo_bound3}
\mathop{\max}_{a_{s},\ldots,a_{2}} ~\sum_{i=1}^{\lceil s/2\rceil} \frac{s}{2i} \cdot \psi_n\left(a_{2i},a_i,a_i\right)
\end{equation}
has a global maximum and the maximum occurs at the expected values of the $a_i, ~\widehat{a}_{i}$ given by
\begin{equation}
\label{eq:prob_1set_ssparse_expo_bound_grad_cp3}
\widehat{a}_{2i} = \widehat{a}_{i}\left(2 - \frac{\widehat{a}_{i}}{n}\right) \quad \mbox{for} \quad i=1,2,4,\ldots,\lceil s/2\rceil,
\end{equation}
which are a solution of the following polynomial system.
\begin{align}
\label{eq:prob_1set_ssparse_expo_bound_grad_cp1}
a_{\lceil s/2\rceil}^2 - 2 n a_{\lceil s/2\rceil} + n a_{s} & = 0,\nonumber\\
a_{2i}^3 - 2a_ia_{2i}^2 + 2 a_i^2 a_{2i} - a_i^2 a_{4i} & = 0, \nonumber\\
\mbox{for} \quad i & = 1,2,\ldots,\lceil s/4\rceil,
\end{align}
where $a_1=d$. If $a_{s}$ is constrained to be less than $\hat{a}_{s}$, then there is a different global maximum,
instead the $a_i$ satisfy the following system
\begin{align}
\label{eq:prob_1set_ssparse_expo_bound_grad_cp2}
a_{2i}^3 - 2a_ia_{2i}^2 + 2 a_i^2 a_{2i} - a_i^2 a_{4i} & = 0, \nonumber \\
 \mbox{for} \quad i & = 1,2,4,\ldots,\lceil s/4\rceil,
\end{align}
again with $a_1=d$.
\end{lemma}

\begin{proof}
Define
\begin{equation}
\label{eq:Psi_n_def}
\widetilde{\Psi}_n\left(a_{s},\ldots,a_{2},d\right) := \sum_{i=1}^{\lceil s/2\rceil} \frac{s}{2i} \cdot \psi_n\left(a_{2i},a_i,a_i\right).
\end{equation}
Using the definition of $\psi_n(\cdot)$ in \eqref{eq:psi_n_def} we therefore have
\begin{multline}
\label{eq:prob_1set_ssparse_expo_bound5}
\widetilde{\Psi}_n\left(a_{s},\ldots,a_{2},d\right) = \sum_{i=1}^{\lceil s/2\rceil} \frac{s}{2i} \cdot\left[a_i\cdot \hbox{H}\left(\frac{a_{2i}-a_i}{a_i}\right)\right. + \\
\left. \left(n-a_i\right) \cdot \hbox{H}\left(\frac{a_{2i}-a_i}{n-a_i}\right) - n\cdot \hbox{H}\left(\frac{a_i}{n}\right) \right].
\end{multline}
The gradient of $\widetilde{\Psi}_n\left(a_{s},\ldots,a_{2},d\right), ~\nabla \widetilde{\Psi}_n\left(a_{s},\ldots,a_{2},d\right)$ is given by
\begin{multline}
\label{eq:prob_1set_ssparse_expo_bound_grad}
\left(\log\left[\frac{\left(2a_{\lceil s/2\rceil} - a_{s}\right) \left(n-a_{s}\right)} {\left(a_{s}-a_{\lceil s/2\rceil}\right)^2}\right],\right.\\ \left. \frac{s}{2i} \cdot \log\left[\frac{a_{2i}\left(a_{4i}-a_{2i}\right) \left(2a_{i}-a_{2i}\right)} {\left(2a_{2i}-a_{4i}\right) \left(a_{2i}-a_{i}\right)^2}\right]\right)^T \\ \mbox{for} \quad i=1,2,4,\ldots,\lceil s/4\rceil,
\end{multline}
where $v^T$ is the transpose of the vector $v$. Obtaining the critical points by solving $\nabla \widetilde{\Psi}_n\left(a_{s},\ldots,a_{2},d\right)=0$ leads to the polynomial system \eqref{eq:prob_1set_ssparse_expo_bound_grad_cp1}.

The Hessian, $\nabla^2 \widetilde{\Psi}_n\left(a_{s},\ldots,a_{2},d\right)$ at these optimal $a_{i}$ which are the solutions to the polynomial system \eqref{eq:prob_1set_ssparse_expo_bound_grad_cp1} is negative definite which implies that this unique critical point is a global maximum point. Let the solution of the system be the $\hat{a}_{i}$ then they satisfy a recurrence formula \eqref{eq:prob_1set_ssparse_expo_bound_grad_cp3} which is equivalent to their expected values as explained in the paragraph that follows.

We estimate the uniformly distributed parameter relating $a_{2i}$ to $a_{i}$. The best estimator of this parameter is the maximum likelihood estimator which we calculate from the maximum log-likelihood estimator (MLE). The summation of the $\psi_n(\cdot)$ is the logarithm of the join density functions for the $a_{2i}$. The MLE is obtained by maximizing this summation and it corresponds to the expected log-likelihood. Therefore, the parameters given implicitly by \eqref{eq:prob_1set_ssparse_expo_bound_grad_cp3} are the expected log-likelihood which implies that the values of the $\hat{a}_{j}$ in \eqref{eq:prob_1set_ssparse_expo_bound_grad_cp3} are the expected values of the $a_{i}$.

If we restrict $a_{s}$ to take a fixed value, then $\nabla \widetilde{\Psi}_n\left(a_{s},\ldots,a_{2},d\right)$ is given by
\begin{multline}
\label{eq:prob_1set_ssparse_expo_bound_grad2}
\left(\frac{s}{2i} \cdot \log\left[\frac{a_{2i}\left(a_{4i}-a_{2i}\right) \left(2a_{i}-a_{2i}\right)} {\left(2a_{2i}-a_{4i}\right) \left(a_{2i}-a_{i}\right)^2}\right]\right)^T \\ \mbox{for} \quad i=1,2,4,\ldots,\lceil s/4\rceil.
\end{multline}
Obtaining the critical points by solving $\nabla \widetilde{\Psi}_n\left(a_{s},\ldots,a_{2},d\right)=0$ leads to the polynomial system \eqref{eq:prob_1set_ssparse_expo_bound_grad_cp2}.

Given $a_{s}$, the Hessian, $\nabla^2 \widetilde{\Psi}_n\left(a_{s},\ldots,a_{2},d\right)$ at these optimal $a_i$ which are the solutions to the polynomial system \eqref{eq:prob_1set_ssparse_expo_bound_grad_cp2} is negative definite which implies that this unique critical point is a global maximum; this case differs from a maximum likelihood estimation because of the extra constraint of fixing $a_{s}$.
\end{proof}

The dyadic splitting technique we employ requires greater care of the polynomial term in the large deviation bound of $\hbox{P}_n\left(x,y,z\right)$ in \eqref{eq:intersect_prob}; Lemma \ref{lem:prob_1set_poly} establishes the polynomial term.
\begin{definition}
\label{def:prob_1set_poly}
$\hbox{P}_n\left(x,y,z\right)$ defined in \eqref{eq:intersect_prob} satisfies
the upper bound
\begin{equation}
\label{eq:prob_1set_poly1}
\hbox{P}_n\left(x,y,z\right)\le \pi \left( x,y,z \right) \exp(\psi_n(x,y,z))
\end{equation}
with bounds of $\pi \left( x,y,z \right)$ given in Lemma \ref{lem:prob_1set_poly}.
\end{definition}

\begin{lemma}
\label{lem:prob_1set_poly}
For $\pi \left( x,y,z \right)$ and $\hbox{P}_n\left(x,y,z\right)$ given by \eqref{eq:prob_1set_poly1} and \eqref{eq:intersect_prob} respectively, if $\{y,z\}<x<y+z$, $\pi \left( x,y,z \right)$ is given by
\begin{equation}
\label{eq:prob_1set_poly3a}
 \left(\frac{5}{4}\right)^4 \left[ \frac{yz(n-y)(n-z)}{2\pi n(y+z-x)(x-y)(x-z)(n-x)} \right]^{\frac{1}{2}},
\end{equation}
otherwise $\pi \left( x,y,z \right)$ has the following cases.
\begin{align}
\label{eq:prob_1set_poly3b}
\left( \frac{5}{4} \right)^3 \left[ \frac{y(n-z)}{n(y-z)} \right]^{\frac{1}{2}} & \quad \mbox{if} \quad x=y>z;\\
\label{eq:prob_1set_poly3c}
\left( \frac{5}{4} \right)^3 \left[ \frac{(n-y)(n-z)}{n(n-y-z)} \right]^{\frac{1}{2}} & \quad \mbox{if} \quad x=y+z;\\
\label{eq:prob_1set_poly3d}
\left( \frac{5}{4} \right)^2 \left[ \frac{2\pi z(n-z)}{n} \right]^{\frac{1}{2}} & \quad \mbox{if} \quad x=y=z.
\end{align}
\end{lemma}

\begin{proof}
The Stirling's inequality \cite{cheney98introduction} below would be used in this proof and other proofs to follow.
\begin{multline}
\label{eq:stirling}
\frac{16}{25}\left(2\pi p (1-p)N\right)^{-\frac{1}{2}}e^{N\mathrm{H}(p)} \leq \binom{N}{Np} \\ \leq \frac{5}{4}\left(2\pi p (1-p)N\right)^{-\frac{1}{2}}e^{N\mathrm{H}(p)}.
\end{multline}

From Definition \ref{def:prob_1set_poly} the quantity $\pi \left( x,y,z \right)$ is the
polynomial portion of the large deviation upper bound.  Within this proof we express this by
\begin{equation}
\label{eq:prob_1set_poly2}
\pi \left( x,y,z \right)  = poly \left[ \binom{y}{y+z-x} \binom{n-y}{x-y} {\binom{n}{z}}^{-1} \right].
\end{equation}
We derive the upper bound $\pi \left( x,y,z \right)$ using the Stirling's inequality. The right inequality of \eqref{eq:stirling} is used to upper bound $\binom{y}{y+z-x}$ and $\binom{n-y}{x-y}$ and the left inequality of \eqref{eq:stirling} is used to lower bound $\binom{n}{z}$. If $\{y,z\}<x<y+z$ the bound is well defined and simplifies to 
\eqref{eq:prob_1set_poly3a}.

If $x=y>z$ \eqref{eq:prob_1set_poly3a} is undefined; however, substituting $y$ for $x$ in \eqref{eq:prob_1set_poly2} gives $\binom{y}{y+z-x}=\binom{y}{z}$ and $\binom{n-y}{x-y}=\binom{n-y}{0}=1$. We upper bound the product $\binom{y}{z} {\binom{n}{z}}^{-1}$ using the right inequality in \eqref{eq:stirling} to bound $\binom{y}{z}$ from above and the left inequality in \eqref{eq:stirling} to bound from below $\binom{n}{z}$. The resulting polynomial part of the product simplifies to \eqref{eq:prob_1set_poly3b}.

If $x=y+z$, then $\binom{y}{y+z-x}=\binom{y}{0}=1$ and $\binom{n-y}{x-y}=\binom{n-y}{z}$. As above, we upper bound the product of $\binom{n-y}{x-y}$ and ${\binom{n}{z}}^{-1}$ using \eqref{eq:stirling} and simplify the polynomial part of this product to get \eqref{eq:prob_1set_poly3c}. If instead $x=y=z$, then $\binom{y}{y+z-x}=\binom{y}{0}$ and $\binom{n-y}{x-y}=\binom{n-y}{0}$ both of which equal 1. Therefore the bound only involves ${\binom{n}{z}}^{-1}$ which we bound using \eqref{eq:stirling} and the resulting polynomial part simplifies to \eqref{eq:prob_1set_poly3d}.
\end{proof}

\begin{corollary}
\label{cor:pi_monotonicity}
If $n>2y$, then $\pi(y,y,y)$ is monotonically increasing in $y$.
\end{corollary}
\begin{proof} If $n>2y$, \eqref{eq:prob_1set_poly3d} implies that $\pi(y,y,y)$ is proportional to $\sqrt{y}$, i.e. $\pi(y,y,y) = c\sqrt{y}$, with $c>0$ and $c\sqrt{y}$ is monotonic in $y$.
\end{proof}

\section{Restricted isometry constants and Compressed Sensing algorithms}\label{sec:ric12}
Here we introduce $\mathrm{RIC}_2$ and briefly discuss the implications of $\mathrm{RIC}_1$ and $\mathrm{RIC}_2$ to compressed sensing algorithms in Section \ref{sec:sec_ric12}. In Section \ref{sec:algorithms} we present the first ever quantitative comparison of the performance guarantees of some of the compressed sensing algorithms proposed for sparse matrices as stated in Definition \ref{def:1set_expansion}.

\subsection{Restricted isometry constants}\label{sec:sec_ric12}
It is possible to include noise in the Compressed Sensing model, for instance $y=Ax+e$ where $e$ is a noise vector capturing the model misfit or the non-sparsity of the signal $x$. The $\ell_0$-minimization problem \eqref{eq:decoding_exact} in the noise case setting is
\begin{equation}
\label{eq:decoding_noise}
 \min_{x \in \chi^N} \|x\|_0 \quad \mbox{subject to} \quad \|Ax - y\|_2<\|e\|_2,
\end{equation}
where $\|e\|_2$ is the magnitude of the noise.

Problems \eqref{eq:decoding_exact} and \eqref{eq:decoding_noise} are in general NP-hard and hence intractable. To benefit from the rich literature of algorithms available in both convex and non-convex optimization the  $\ell_0$-minimization problem is relaxed to an $\ell_p$-minimization one for $0<p\leq1$. It is well known that the $\ell_p$ norm for $0<p\leq1$ are sparsifying norms, see \cite{candes2005decoding,foucart2009sparsest}.  In addition, there are specifically designed classes of algorithms that take on the $\ell_0$ problem and they have been referred to as greedy algorithms. When using dense sensing matrices, $A$, popular greedy algorithms include Normalized Iterative Hard Thresholding (NIHT), \cite{blumensath2010niht}, Compressive Sampling Matching Pursuits (CoSAMP), \cite{needell2009cosamp}, and Subspace Pursuit (SP), \cite{dai2009subspace}. When $A$ is sparse and non-mean zero, a different set of {\em combinatorial} greedy algorithms have been proposed which iteratively locates and eliminate large (in magnitude) components of the vector, \cite{berinde2008combining}. They include Expander Matching Pursuit (EMP), \cite{indyk2008near}, Sparse Matching Pursuit (SMP), \cite{berinde2008practical}, Sequential Sparse Matching Pursuit (SSMP), \cite{berinde2009sequential}, Left Degree Dependent Signal Recovery (LDDSR), \cite{xu2007efficient}, and Expander Recovery (ER), \cite{jafarpour2008efficient,jafarpour2009efficient}.

The convergence analysis of nearly all of these algorithms rely heavily on restricted isometry constants (RIC). As we saw earlier RICs measures how near isometry  $A$ is when applied to $k$-sparse vectors in some norm. For the $\ell_1$ norm, also known as the Manhattan norm, $\mathrm{RIC}_1$ is stated in \eqref{eq:ric1}.  The restricted Euclidian norm isometry, introduced by Cand\`es in \cite{candes2008restricted}, is denoted by $\mathrm{RIC}_2$ and is defined in Definition \ref{def:ric2}.
\begin{definition}[$\mathrm{RIC}_2$]
\label{def:ric2}
  Define $\chi^N$ to be the set of all $k-$sparse vectors and draw  an $n\times N$ matrix $A$, then for all $x\in\chi^N$, $A$ has $\mathrm{RIC}_2$, with lower and upper $\mathrm{RIC}_2$, $L(k,n,N;A)$ and $U(k,n,N;A)$ respectively, when the following holds.
\begin{multline*}
 \left(1-L(k,n,N;A)\right)\|x\|_2 \leq \|Ax\|_2 \\ \leq \left(1+U(k,n,N;A)\right)\|x\|_2.
\end{multline*}
\end{definition}

The computation of $\mathrm{RIC}_1$ for adjacency matrices of lossless $(k,d,\e)$-expander graphs is equivalent to calculating $\e$.  The computation of $\mathrm{RIC}_2$ is intractable except for trivially small problem sizes $(k,n,N)$ because it involves doing a combinatorial search over all $\binom{N}{k}$ column submatrices of $A$. As a results attempts have been made to derive $\mathrm{RIC}_2$ bounds. Some of these attempts have been successful in deriving $\mathrm{RIC}_2$ bounds for the Gaussian ensemble and these bounds have evolved from the first by Cand\`es and Tao in \cite{candes2005decoding}, improved by Blanchard, Cartis and Tanner in \cite{blanchard2011compressed} and further improved by Bah and Tanner in \cite{bah2010improved}.

$\mathrm{RIC}_2$ bounds have been used to derive sampling theorems for compressed sensing algorithms - $\ell_1$-minimization and the greedy algorithms for dense matrices, NIHT, CoSAMP, and SP. Using the phase transition framework with $\mathrm{RIC}_2$ bounds Blanchard et. al. compared performance of these algorithms in \cite{blanchard2011phase}. In a similar vain, as another key contribution of this paper we provide sampling theorems for $\ell_1$-minimization and combinatorial greedy algorithms, EMP, SMP, SSMP, LDDSR and ER, proposed for SE and SSE matrices.

\subsection{Algorithms and their performance guarantees}\label{sec:algorithms}

Theoretical guarantees have been given for $\ell_1$ recovery and other greedy algorithms including EMP, SMP, SSMP, LDDSR and ER designed to do compressed sensing recovery with adjacency matrices of lossless expander graphs and by extension SSE matrices. Sparse matrices have been observed to have recovery properties comparable to dense matrices for $\ell_1$-minimization and some of the aforesaid algorithms, see \cite{berinde2008combining,berinde2008sparse,jafarpour2008efficient,xu2007further,xu2007efficient} and the references therein. Base on theoretical guarantees, we derived sampling theorems and present here phase transition curves which are plots of phase transition functions $\r^{alg}(\d;d,\e)$ of algorithms such that for $k/n \rightarrow \r < (1-\g)\r^{alg}(\d;d,\e), ~\g>0$, a given algorithm is guaranteed to recover all $k$-sparse signals with overwhelming probability approaching one exponentially in $n$.

\subsubsection{$\ell_1$-minimization}\label{sec:ell1_guarantee}
Note that $\ell_1$-minimization is not an algorithm per se, but can be solved using Linear Programming (LP) algorithms.  Berinde et. al. showed in \cite{berinde2008combining} that $\ell_1$-minimization can be used to perform signal recovery with binary matrices coming from expander graphs. We reproduce the formal statement of this guarantee in the following theorem, the proof of which can be found in \cite{berinde2008combining,berinde2008sparse}.


\begin{theorem}[Theorem 3, \cite{berinde2008combining}, Theorem 1, \cite{berinde2008sparse}]
\label{thm:ell1_guarantee}
Let $A$ be an adjacency matrix of a lossless $(k,d,\e)$-expander graph with $\alpha(\epsilon) = 2\epsilon/(1-2\epsilon)<1/2$.
Given any two vectors $x$, $\hat{x}$ such that $Ax = A\hat{x}$, and $||\hat{x}||_1 \leq ||x||_1$, let $x_k$ be the largest (in magnitude) coefficients of $x$, then
\begin{equation}
\label{eq:ell1_guarantee}
 ||x-\hat{x}||_1 \leq \frac{2}{1-2\alpha(\epsilon)}||x-x_k||_1.
\end{equation}
\end{theorem}

The condition that $\alpha(\epsilon) = 2\epsilon/(1-2\epsilon)<1/2$ implies the sampling theorem stated as Corollary \ref{cor:ell1_guarantee}, that when satisfied ensures a positive upper bound in \eqref{eq:ell1_guarantee}.  The resulting sampling theorem is given by $\r^{\ell_1}(\d;d,\e)$ using $\e=1/6$ from Corollary \ref{cor:ell1_guarantee}.

\begin{corollary}[\cite{berinde2008combining}]
\label{cor:ell1_guarantee}
$\ell_1$-minimization is guaranteed to recover any $k$-sparse vector from its linear measurement by an adjacency matrix of a lossless $(k,d,\e)$-expander graph with $\epsilon < 1/6$.
\end{corollary}
\begin{proof} Setting the denominator of the fraction in the right hand side of \eqref{eq:ell1_guarantee} to be greater than zero gives the required results. \end{proof}

\subsubsection{Sequential Sparse Matching Pursuit (SSMP)}
Introduced by Indyk and Ruzic in \cite{berinde2009sequential}, SSMP has evolved as an improvement of Sparse Matching Pursuit (SMP) which was an improvement on Expander Matching Pursuit (EMP). EMP also introduced by Indyk and Ruzic in \cite{indyk2008near} 
uses a voting-like mechanism to identify and eliminate large (in magnitude) components of signal. EMP's drawback is that the {\em empirical} number of measurements it requires to achieve correct recovery is suboptimal. SMP, introduced by Berinde, Indyk and Ruzic in \cite{berinde2008practical},  improved on the drawback of EMP. However, it's original version had convergence problems when the input parameters ($k$ and $n$) fall outside the theoretically guaranteed region. This is fixed by the SMP package which forces convergence when the user provides an additional convergence parameter. In order to correct the aforementioned problems of EMP and SMP, Indyk and Ruzic developed SSMP. It is a version of SMP where updates are done sequentially instead of parallel, consequently convergence is automatically achieved. All three algorithms have the same theoretical recovery guarantees, which we state in Theorem \ref{thm:ssmp_guarantee}, but SSMP has better empirical performances compared to it's predecessors.

Algorithm 1 below is a pseudo-code of the SSMP algorithm based on the following problem setting. The measurement matrix $A$ is an $n\times N$ adjacency matrix of a lossless $((c+1)k,d,\e/2)$-expander scaled by $d$ and $A$ has a lower $\mathrm{RIC}_1$, $L\left((c+1)k,n,N\right) = \e$. The measurement vector $y = Ax + e$ where $e$ is a noise vector and $\eta = \|e\|_1$. We denote by $H_k (y)$ the hard thresholding operator which sets to zero all but the largest, in magnitude, $k$ entries of $y$.

\begin{table}[h]
\centering
\begin{tabular}{l}
\hline
 \textbf{Algorithm 1} Sequential Sparse Matching Pursuit (SSMP) \cite{berinde2009sequential}\\
 \hline
 \textbf{Input:} $A, ~y, ~\eta$\\
 \textbf{Output:} $k$-sparse approximation $\hat{x}$ of the target signal $x$\\
 \hline
 \textbf{Initialization:}\\
    \quad $1.$ Set $j=0$ \\
    \quad $2.$ Set $x_j=0$ \\
 \textbf{Iteration:} Repeat $T = \mathcal{O}\left(\log\left(\|x\|_1/\eta\right)\right)$ times\\
    \quad $1.$ Set $j=j+1$ \\
    \quad $2.$ Repeat $(c-1)k$ times \\
    \qquad $a)$ Find a coordinate $i$ \& an increment $z$ that \qquad \qquad\\
    \quad \quad \quad ~minimizes $\|A\left(x_j + ze_i\right) - y\|_1$ \\
    \qquad $b)$ Set $x_j$ to $x_j + ze_i$\\
    \quad $3.$ Set $x_j = H_k\left(x_j\right)$ \\
 \textbf{Return} $\hat{x} = x^T$ \\
 \hline
\end{tabular}
\label{tab:ssmp_pcode}
\end{table}

The recovery guarantees for SSMP (also for EMP and SMP) are formalized by the following theorem from which we deduce the recovery condition (sampling theorem) in terms of $\e$ in Corollary \ref{cor:ssmp_guarantee}. Based on Corollary \ref{cor:ssmp_guarantee} deduced from Theorem \ref{thm:ssmp_guarantee} we derived phase transition, $\r^{SSMP}(\d;d,\e)$, for SSMP.

\begin{theorem}[Theorem 10, \cite{indyk2008near}]
\label{thm:ssmp_guarantee}
Let $A$ be an adjacency matrix of a lossless $(k,d,\e)$-expander graph with $\epsilon<1/16$.  Given a vector $y = Ax + e$, the algorithm returns approximation vector $\hat{x}$ satisfying
\begin{equation}
\label{eq:ssmp_guarantee}
||x-\hat{x}||_1 \leq \frac{1-4\epsilon}{1-16\epsilon}||x-x_k||_1 + \frac{6}{(1-16\epsilon)d}||e||_1,
\end{equation}
where $x_k$ is the $k$ largest (in magnitude) coordinates of $x$.
\end{theorem}

\begin{corollary}[\cite{indyk2008near}]
\label{cor:ssmp_guarantee}
SSMP, EMP, and SMP are all guaranteed to recover any $k$-sparse vector from its linear measurement by an adjacency matrix of a lossless $(k,d,\e)$-expander graph with $\epsilon < 1/16$.
\end{corollary}

\subsubsection{Expander Recovery (ER)}
Introduced by Jafarpour et. al. in \cite{jafarpour2008efficient,jafarpour2009efficient}, ER is an improvement on an earlier algorithm introduced by Xu and Hassibi in \cite{xu2007efficient} known as Left Degree Dependent Signal Recovery (LDDSR). The improvement was mainly on the number of iterations used by the algorithms and the type of expanders used, from $(k,d,1/4)$-expanders for LDDSR to $(k,d,\epsilon)$-expander for any $\epsilon < 1/4$ for ER. Both algorithms use this concept of a {\em gap} defined below.

\begin{definition}[gap, \cite{xu2007efficient,jafarpour2008efficient,jafarpour2009efficient}]
\label{def:gap}
Let $x$ be the original signal and $y = Ax.$ Furthermore, let $\hat{x}$ be our estimate for $x$. For each value $y_i$ we define a gap $g_i$ as:
\begin{equation}
\label{eq:gap}
g_i = y_i - \sum_{j=1}^{N} A_{ij}\hat{x}_j.
\end{equation}
\end{definition}

Algorithm 2 below is a pseudo-code of the ER algorithm for an original $k$-sparse signal $x\in \mathbb{R}^N$ and the measurements $y = Ax$ with an ${n\times N}$ measurement matrix $A$ that is an adjacency matrix of a lossless $(2k,d,\e)$-expander and $\e< 1/4$.  The measurements are assumed to be without noise, so we aim for exact recovery. The authors of \cite{jafarpour2008efficient,jafarpour2009efficient} have a modified version of the algorithm for when $x$ is almost $k$-sparse.

\begin{table}[h]
\centering
\begin{tabular}{l}
\hline
 \textbf{Algorithm 2} Expander Recovery (ER) \cite{jafarpour2008efficient,jafarpour2009efficient}\\
 \hline
 \textbf{Input:} $A, ~y$\\
 \textbf{Output:} $k$-sparse approximation $\hat{x}$ of the original signal $x$\\
 \hline
 \textbf{Initialization:}\\
    \quad $1.$ Set $\hat{x}=0$ \\
 \textbf{Iteration:} Repeat at most $2k$ times\\
    \quad $1.$ \textbf{if} $y = A\hat{x}$ \textbf{then}\\
    \quad $2.$ \quad \textbf{return} $\hat{x}$ and exit \\
    \quad $3.$ \textbf{else}\\
    \quad $4.$ \quad Find a variable node $\hat{x}_j$ such that at least $(1-2\e)d$ of the\\
    \qquad \quad measurements it participated in, have identical gap $g$\\
    \quad $5.$ \quad Set $\hat{x}_j = \hat{x}_j + g$, and go to 2. \\
    \quad $6.$ \textbf{end if}\\
 \hline
\end{tabular}
\label{tab:er_pcode}
\end{table}

Theorem \ref{thm:ER} gives recovery guarantees for ER. Directly from this theorem we read-off the recovery condition in terms of $\e$ for Corollary \ref{cor:ER}, from which
we derive phase transition functions, $\r^{ER}(\d;d,\e)$, for ER.
\begin{theorem}[Theorem 6, \cite{jafarpour2009efficient}]
\label{thm:ER}
Let $A\in \mathbb{R}^{n\times N}$ be the adjacency matrix of a lossless $(2k,d,\epsilon)$-expander graph, where $\epsilon < 1/4$ and $n = \mathcal{O}\left(k\log(N/k)\right)$. Then, for any $k$-sparse signal $x$, given $y = Ax$, ER recovers $x$ successfully in at most $2k$ iterations.
\end{theorem}

\begin{corollary}
\label{cor:ER}
ER is guaranteed to recover any $k$-sparse vector from its linear measurement by an adjacency matrix of a lossless $(k,d,\e)$-expander graph with $\epsilon < 1/4$.
\end{corollary}

\subsubsection{Comparisons of phase transitions of algorithms}
\begin{figure}[h]
\centering
\includegraphics[width=0.5\textwidth]{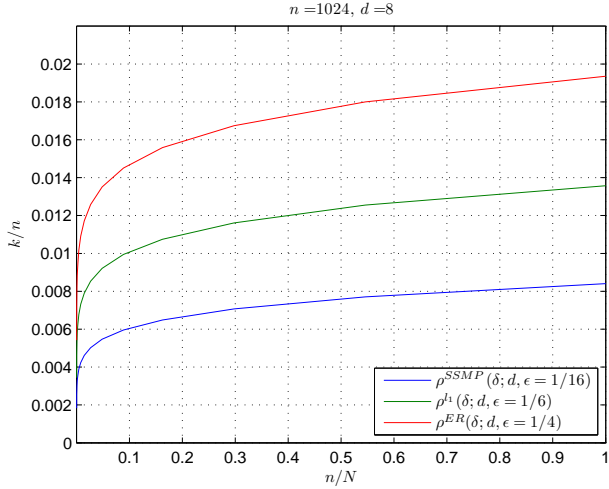}
\caption{Phase transition curves $\rho^{alg}\left(\delta;d,\e\right)$ computed over finite values of $\delta \in (0,1)$ with $d$ fixed and the different $\epsilon$ values for each algorithm - 1/4, 1/6 and 1/16 for ER, $\ell_1$ and SSMP respectively.}
\label{fig:pt_alg_compare}
\end{figure}

Fig. \ref{fig:pt_alg_compare} compares the phase transition plot of $\r^{SSMP}(\d;d,\e)$ for SSMP (also for EMP and SMP), the phase transition of plot $\r^{ER}(\d;d,\e)$ for ER (also of LDDSR) and the phase transition plot of $\r^{\ell_1}(\d;d,\e)$ for $\ell_1$-minimization. Remarkably, for ER and LDDSR recovery is guaranteed for a larger portion of the $(\delta,\rho)$ plane than is guaranteed by the theory for $\ell_1$-minimization using sparse matrices; however,  $\ell_1$-minimization has a larger recovery region than does SSMP, EMP, and SMP.

\begin{figure}[h]
\centering
\includegraphics[width=0.5\textwidth]{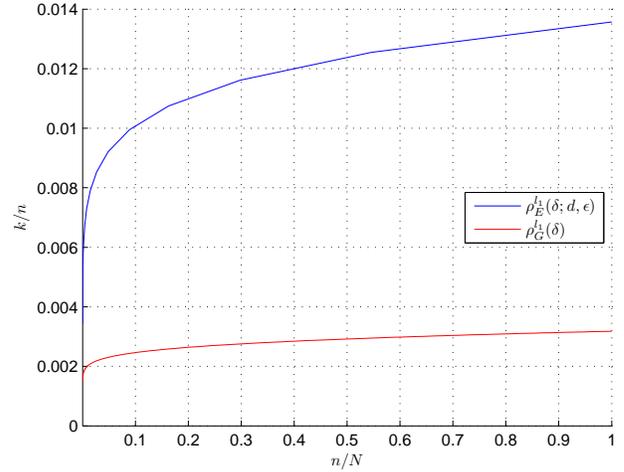}
\caption{Phase transition plots of $\ell_1$, $\rho^{\ell_1}_G\left(\delta\right)$, for Gaussian matrices derived using $\mathrm{RIC}_2$ and $\rho^{\ell_1}_E\left(\delta;d,\e\right)$ for adjacency matrices of expander graphs with $n=1024$, $d=8$, and $\e=1/6$.}
\label{fig:pt_l1_compare}
\end{figure}

Fig. \ref{fig:pt_l1_compare} shows a comparison of the phase transition of $\ell_1$-minimization as presented by Blanchard et. al. in \cite{blanchard2011phase} for dense Gaussian matrices based on $\mathrm{RIC}_2$ analysis and the phase transition we derived here for the sparse binary matrices coming from lossless expander based on $\mathrm{RIC}_1$ analysis. This shows a remarkable difference between the two with sparse matrices having better performance guarantees. However, these improved recovery guarantees are likely more due to the closer match of the method of analysis than to the efficacy of sparse matrices over dense matrices.


\section{Proof of mains results}\label{sec:proof}

\vspace{3mm}

\subsection{Proof of Theorem \ref{thm:prob_bound_1set_expansion}}\label{sec:main_proof}
By the dyadic splitting $\left|A_s\right|=\left|A_{\lceil\frac{s}{2}\rceil}^1\cup A_{\lfloor\frac{s}{2}\rfloor}^2\right|$ and therefore
\begin{align}
\label{eq:prob_1set_ssparse1}
& \hbox{Prob}\left(\left|A_s\right| \leq a_s\right) = \hbox{Prob}\left(\left|A_{\lceil\frac{s}{2}\rceil}^1\cup A_{\lfloor\frac{s}{2}\rfloor}^2\right| \leq a_s\right)\\
\label{eq:prob_1set_ssparse2}
& = \sum_{l_s} \sum_{l_{\lceil\frac{s}{2}\rceil},l_{\lfloor\frac{s}{2}\rfloor}} \hbox{Prob}\left(\left|A_{\lceil\frac{s}{2}\rceil}^1\cup A_{\lfloor\frac{s}{2}\rfloor}^2\right| = l_s\right)\\
\label{eq:prob_1set_ssparse3}
& = \sum_{l_s} \sum_{l_{\lceil\frac{s}{2}\rceil}^1} \sum_{l_{\lfloor\frac{s}{2}\rfloor}^2} \hbox{P}_n\left(l_s,l_{\lceil\frac{s}{2}\rceil}^1,l_{\lfloor\frac{s}{2}\rfloor}^2\right) \times \nonumber\\
& \qquad \hbox{Prob}\left(\left|A_{\lceil\frac{s}{2}\rceil}^1\right| = l_{\lceil\frac{s}{2}\rceil}^1\right) \hbox{Prob}\left(\left|A_{\lfloor\frac{s}{2}\rfloor}^2\right| = l_{\lfloor\frac{s}{2}\rfloor}^2\right).
\end{align}
From \eqref{eq:prob_1set_ssparse1} to \eqref{eq:prob_1set_ssparse2} we sum over all possible events while from \eqref{eq:prob_1set_ssparse2} to \eqref{eq:prob_1set_ssparse3}, in line with the splitting technique, we simplify the probability to the product of the probabilities of the cardinalities of $\left|A_{\lceil\frac{s}{2}\rceil}^1\right|$ and $\left|A_{\lfloor\frac{s}{2}\rfloor}^2\right|$ and their intersection.

In a slight abuse of notation we write $\mathop{\sum_{l^j}}_{j=1,\ldots,x}$ to denote applying the sum $x$ times. Now we use Lemma \ref{lem:num_size_split} to simplify \eqref{eq:prob_1set_ssparse3} as follows.
\begin{multline}
\label{eq:prob_1set_ssparse4}
\mathop{\sum_{l^{j_1}_{Q_0}}}_{j_1=1,\ldots,q_0} \mathop{\sum_{l^{j_2}_{Q_1}}}_{j_2=1,\ldots,q_1} \mathop{\sum_{l^{j_3}_{R_1}}}_{j_3=1,\ldots,r_1} \hbox{P}_n\left(l^{j_1}_{Q_0},l^{2j_1-1}_{\lceil\frac{Q_0}{2}\rceil},l^{2j_1}_{\lfloor\frac{Q_0}{2}\rfloor}\right) \\
\times \prod_{j_2 = 1}^{q_1}\hbox{Prob}\left(\left|A^{j_2}_{Q_1}\right| = l^{j_2}_{Q_1}\right) \\ \times \prod_{j_3=q_1+1}^{q_1+r_1} \hbox{Prob}\left(\left|A^{j_3}_{R_1}\right| = l^{j_3}_{R_1}\right).
\end{multline}

Let's quickly verify that \eqref{eq:prob_1set_ssparse4} is the same as \eqref{eq:prob_1set_ssparse3}. By Lemma \ref{lem:num_size_split}, $Q_0 = s$ is the number of columns in the set at the zeroth level of the split while $q_0 = 1$ is the number of sets with $Q_0$ columns at the zeroth level of the split. Thus for $j_1 = 1$ the first summation and the $\hbox{P}_n(\cdot)$ term are the same in the two equations. If $\lceil\frac{Q_0}{2}\rceil = \lfloor\frac{Q_0}{2}\rfloor$, then they are both equal to $Q_1$ and $q_1=2$ while $r_1 =0$. If on the other hand $\lceil\frac{Q_0}{2}\rceil = \lfloor\frac{Q_0}{2}\rfloor + 1$, then $q_1=1$ and $r_1=1$. In either case we have the remaining part of the expression of \eqref{eq:prob_1set_ssparse3} i.e. the second two summations and the product of the two $\hbox{Prob}(\cdot)$.

Now we proceed with the splitting - note \eqref{eq:prob_1set_ssparse4} stopped only at the first level. At the next level, the second, we will have $q_2$ sets with $Q_2$ columns and $r_2$ sets with $R_2$ columns which leads to the following expression.
\begin{multline}
\label{eq:prob_1set_ssparse5}
\mathop{\sum_{l^{j_1}_{Q_0}}}_{j_1=1,\ldots,q_0} \mathop{\sum_{l^{j_2}_{Q_1}}}_{j_2=1,\ldots,q_1} \mathop{\sum_{l^{j_3}_{R_1}}}_{j_3=1,\ldots,r_1} \hbox{P}_n\left(l^{j_1}_{Q_0},l^{2j_1-1}_{\lceil\frac{Q_0}{2}\rceil},l^{2j_1}_{\lfloor\frac{Q_0}{2}\rfloor}\right) \\ \times \Bigg{[} \mathop{\sum_{l^{j_4}_{Q_2}}}_{j_4=1,\ldots,q_2} \mathop{\sum_{l^{j_5}_{R_2}}}_{j_5=1,\ldots,r_2}
\hbox{P}_n\left(l^{j_2}_{Q_1},l^{2j_2-1}_{\lceil\frac{Q_1}{2}\rceil},l^{2j_2}_{\lfloor\frac{Q_1}{2}\rfloor}\right) \\ \hbox{P}_n\left(l^{j_3}_{R_1},l^{2j_3-1}_{\lceil\frac{R_1}{2}\rceil},l^{2j_3}_{\lfloor\frac{R_1}{2}\rfloor}\right)
\times \prod_{j_4 = 1}^{q_2}\hbox{Prob}\left(\left|A^{j_4}_{Q_1}\right| = l^{j_4}_{Q_1}\right) \\ \prod_{j_5=q_2+1}^{q_2+r_2} \hbox{Prob}\left(\left|A^{j_5}_{R_1}\right| = l^{j_5}_{R_1}\right)\Bigg{]}.
\end{multline}

We continue this splitting of each instance of $\hbox{Prob}(\cdot)$ for $\lceil \log_2 s\rceil + 1$ levels (from level $0$ to level $\lceil\log_2s\rceil$) until reaching sets with single columns. Note that for the resulting binary tree from the splitting it suffice to stop at level $\lceil\log_2s\rceil-1$ as can be seen in Fig. \ref{fig:splitting}, since at this level we have all the information required for the enumeration of the sets. By construction, at level $\lceil\log_2s\rceil$, omitted from Fig. \ref{fig:splitting}, the probability that the single column has $d$ nonzeros is one.  This process gives a complicated product of nested sums of $\hbox{P}_n(\cdot)$ which we express as

\begin{multline}
\label{eq:prob_1set_ssparse6}
\mathop{\sum_{l^{j_1}_{Q_0}}}_{j_1=1,\ldots,q_0} \mathop{\sum_{l^{j_2}_{Q_1}}}_{j_2=1,\ldots,q_1} \mathop{\sum_{l^{j_3}_{R_1}}}_{j_3=1,\ldots,r_1} \hbox{P}_n\left(l^{j_1}_{Q_0},l^{2j_1-1}_{\lceil\frac{Q_0}{2}\rceil},l^{2j_1}_{\lfloor\frac{Q_0}{2}\rfloor}\right) \\
\times \Bigg{[} \mathop{\sum_{l^{j_4}_{Q_2}}}_{j_4=1,\ldots,q_2} \mathop{\sum_{l^{j_5}_{R_2}}}_{j_5=1,\ldots,r_2}
\hbox{P}_n\left(l^{j_2}_{Q_1},l^{2j_2-1}_{\lceil\frac{Q_1}{2}\rceil},l^{2j_2}_{\lfloor\frac{Q_1}{2}\rfloor}\right) \\ \times \hbox{P}_n\left(l^{j_3}_{R_1},l^{2j_3-1}_{\lceil\frac{R_1}{2}\rceil},l^{2j_3}_{\lfloor\frac{R_1}{2}\rfloor}\right) \cdot \Bigg{[} \ldots \Bigg{[} \mathop{\sum_{l^{j_{2\lceil\log_2s\rceil-2}}_{Q_{\lceil\log_2s\rceil-1}}}}_{j_{2\lceil\log_2s\rceil-2} = 1,\ldots,q_{j_{\lceil\log_2s\rceil-1}}}  \\ \hbox{P}_n\left(l^{j_{2\lceil\log_2s\rceil-4}}_{4}, l^{2j_{2\lceil\log_2s\rceil-4}-1}_{2}, l^{2j_{2\lceil\log_2s\rceil-4}}_{2}\right) \\ \times \hbox{P}_n\left(l^{j_{2\lceil\log_2s\rceil-3}}_{3}, l^{2j_{2\lceil\log_2s\rceil-3}-1}_{2}, d\right) \\\times  \hbox{P}_n\left(l^{j_{2\lceil\log_2s\rceil-2}}_{2}, d, d\right)\Bigg{]} \ldots \Bigg{]}.
\end{multline}

Using the definition of $\hbox{P}_n(\cdot)$ in Lemma \ref{lem:intersect_prob} we bound \eqref{eq:prob_1set_ssparse6} by bounding each $\hbox{P}_n(\cdot)$ as in \eqref{eq:prob_1set_poly1} with a product of a polynomial, $\pi(\cdot)$, and an exponential with exponent $\psi_n(\cdot)$.
\begin{multline}
\label{eq:prob_1set_ssparse7}
\mathop{\sum_{l^{j_1}_{Q_0}}}_{j_1=1,\ldots,q_0} \mathop{\sum_{l^{j_2}_{Q_1}}}_{j_2=1,\ldots,q_1} \mathop{\sum_{l^{j_3}_{R_1}}}_{j_3=1,\ldots,r_1} \pi\left(l^{j_1}_{Q_0},l^{2j_1-1}_{\lceil\frac{Q_0}{2}\rceil},l^{2j_1}_{\lfloor\frac{Q_0}{2}\rfloor}\right) \times \\ e^{ \psi_n\left(l^{j_1}_{Q_0},l^{2j_1-1}_{\lceil\frac{Q_0}{2}\rceil},l^{2j_1}_{\lfloor\frac{Q_0}{2}\rfloor}\right)}
\cdot\Bigg{[} \mathop{\sum_{l^{j_4}_{Q_2}}}_{j_4=1,\ldots,q_2} \mathop{\sum_{l^{j_5}_{R_2}}}_{j_5=1,\ldots,r_2} \\
\pi\left(l^{j_2}_{Q_1},l^{2j_2-1}_{\lceil\frac{Q_1}{2}\rceil},l^{2j_2}_{\lfloor\frac{Q_1}{2}\rfloor}\right) \cdot e^{\psi_n\left(l^{j_2}_{Q_1},l^{2j_2-1}_{\lceil\frac{Q_1}{2}\rceil},l^{2j_2}_{\lfloor\frac{Q_1}{2}\rfloor}\right)} \times \\
\pi\left(l^{j_3}_{R_1},l^{2j_3-1}_{\lceil\frac{R_1}{2}\rceil},l^{2j_3}_{\lfloor\frac{R_1}{2}\rfloor}\right) \cdot e^{\psi_n\left(l^{j_3}_{R_1},l^{2j_3-1}_{\lceil\frac{R_1}{2}\rceil},l^{2j_3}_{\lfloor\frac{R_1}{2}\rfloor}\right)} \\ \times \Bigg{[} \ldots \times \Bigg{[}
\mathop{\sum_{l^{j_{2\lceil\log_2s\rceil-2}}_{Q_{\lceil\log_2s\rceil-1}}}}_{j_{2\lceil\log_2s\rceil-2} = 1,\ldots,q_{j_{\lceil\log_2s\rceil-1}}} \\ \pi\left(l^{j_{2\lceil\log_2s\rceil-4}}_{4}, l^{2j_{2\lceil\log_2s\rceil-4}-1}_{2}, l^{2j_{2\lceil\log_2s\rceil-4}}_{2}\right) \\ \times e^{\psi_n\left(l^{j_{2\lceil\log_2s\rceil-4}}_{4}, l^{2j_{2\lceil\log_2s\rceil-4}-1}_{2}, l^{2j_{2\lceil\log_2s\rceil-4}}_{2}\right)} \\ \times \pi\left(l^{j_{2\lceil\log_2s\rceil-3}}_{3}, l^{2j_{2\lceil\log_2s\rceil-3}-1}_{2}, d\right) \\ \times e^{\psi_n\left(l^{j_{2\lceil\log_2s\rceil-3}}_{3}, l^{2j_{2\lceil\log_2s\rceil-3}-1}_{2}, d\right)} \times \\ \pi\left(l^{j_{2\lceil\log_2s\rceil-2}}_{2}, d, d\right) \cdot e^{\psi_n\left(l^{j_{2\lceil\log_2s\rceil-2}}_{2}, d, d\right)}\bigg{]} \ldots \Bigg{]}.
\end{multline}

Using Lemma \ref{lem:psi_n_behaviour} we maximize the $\psi_n(\cdot)$ and hence the exponentials. If we maximize each by choosing $l_{(\cdot)}$ to be $a_{(\cdot)}$, then we can pull the exponentials out of the product. The exponential will then have the exponent $\Psi_n\left(a_s,\ldots,a_2,d\right)$. The factor involving the $\pi(\cdot)$ will be called $\Pi\left(l_s,\ldots,l_2,d\right)$ and we have the following upper bound for \eqref{eq:prob_1set_ssparse7}.
\begin{equation}
\label{eq:prob_1set_ssparse8}
\Pi\left(l_s,\ldots,l_2,d\right) \cdot \exp\left[\Psi_n\left(a_s,\ldots,a_2,d\right)\right],
\end{equation}
where the exponent $\Psi_n\left(a_s,\ldots,a_2,d\right)$ is given by
\begin{equation}
\label{eq:prob_1set_ssparse_exp1}
\psi_n\left(a_{Q_0},a_{\lceil\frac{Q_0}{2}\rceil},a_{\lfloor\frac{Q_0}{2}\rfloor}\right) + \ldots + \psi_n\left(a_{2},d,d\right).
\end{equation}

Now we attempt to bound the probability of interest in \eqref{eq:prob_1set_ssparse1}. This task reduces to bounding $\Pi\left(l_s,\ldots,l_2,d\right)$ and $\Psi_n\left(a_s,\ldots,a_2,d\right)$ in \eqref{eq:prob_1set_ssparse8} and we start with the former, i.e. bounding $\Pi\left(l_s,\ldots,l_2,d\right)$. We bound each sum of $\pi(\cdot)$ in $\Pi\left(l_s,\ldots,l_2,d\right)$ of \eqref{eq:prob_1set_ssparse8} by the maximum of summations multiplied by the number of terms in the sum. From \eqref{eq:prob_1set_poly3d} we see that $\pi(\cdot)$ is maximized when all the three arguments are the same and using Corollary \ref{cor:pi_monotonicity} we take largest possible arguments that are equal in the range of the summation. In this way the following proposition provides the bound we end up.

\begin{proposition}
\label{pro:poly_bound}
 Let's make each summation over the sets with the same number of columns to have the same range where the range we take are the maximum possible for each such set. Let's also maximize $\pi(\cdot)$ where all its three input variables are equal and are equal to the maximum of the third variable. Then we bound each sum by the largest term in the sum multiplied by the number of terms. This scheme combined with Lemma \ref{lem:num_size_split} give the following upper bound on $\Pi\left(l_s,\ldots,l_2,d\right)$.
\begin{multline}
\label{eq:prob_1set_ssparse_poly1}
\left(\bigg{\lceil}\frac{Q_0}{2}\bigg{\rceil}d\left(\frac{5}{4}\right)^2 \sqrt{2\pi\bigg{\lfloor}\frac{Q_0}{2}\bigg{\rfloor}d}\right)^{q_0} \times\\
\prod_{j=1}^{\lceil\log_2s\rceil-2} \left[\left(\bigg{\lceil}\frac{Q_j}{2}\bigg{\rceil}d\left(\frac{5}{4}\right)^2 \sqrt{2\pi\bigg{\lfloor}\frac{Q_j}{2}\bigg{\rfloor}d}\right)^{q_j}\right. \times \\ \left.\left(\bigg{\lceil}\frac{R_j}{2}\bigg{\rceil}d\left(\frac{5}{4}\right)^2 \sqrt{2\pi\bigg{\lfloor}\frac{R_j}{2}\bigg{\rfloor}d}\right)^{r_j} \right] \times \\
\left(\bigg{\lceil}\frac{Q_{\lceil\log_2s\rceil-1}}{2}\bigg{\rceil}d\left(\frac{5}{4}\right)^2 \sqrt{2\pi\bigg{\lfloor}\frac{Q_{\lceil\log_2s\rceil-1}}{2}\bigg{\rfloor}d}\right)^{q_{\lceil\log_2s\rceil-1}}
\end{multline}
\end{proposition}

\begin{proof}
From \eqref{eq:prob_1set_poly3d} we have
\begin{equation}
\label{eq:psi_n_bound}
\pi(y,y,y) = \left(\frac{5}{4} \right)^2 \sqrt{\frac{2\pi y(n-y)}{n}} < \left(\frac{5}{4} \right)^2 \sqrt{2\pi y}.
\end{equation}
Simply put, we bound $\sum_{x} \pi(x,y,z)$ by multiplying the maximum of $\pi(x,y,z)$ with the number of terms in the summation. Remember the order of magnitude of the arguments of $\pi(x,y,z)$ is $x\geq y\geq z$. Therefore, the maximum of $\pi(x,y,z)$ occurs when the arguments are all equal to the maximum value of $z$. In our splitting scheme the maximum possible value of $l_{\big{\lfloor}\frac{Q_j}{2}\big{\rfloor}}$ is $\big{\lfloor}\frac{Q_j}{2}\big{\rfloor}\cdot d$ since there are $d$ nonzeros in each column. Also $l_{\big{\lfloor}\frac{Q_j}{2}\big{\rfloor}}\leq l_{Q_j} \leq l_{\big{\lfloor}\frac{Q_j}{2}\big{\rfloor}} + l_{\big{\lceil}\frac{Q_j}{2}\big{\rceil}}$ so the number of terms in the summation over $l_{Q_j}$ is $\big{\lceil}\frac{Q_j}{2}\big{\rceil}\cdot d$, and similarly for $R_j$. We know the values of the $Q_j$ and the $R_j$ and their quantities $q_j$ and $r_j$ respectively from Lemma \ref{lem:num_size_split}.

We replace $y$ by $\big{\lfloor}\frac{Q_j}{2}\big{\rfloor}\cdot d$ or $\big{\lfloor}\frac{R_j}{2}\big{\rfloor}\cdot d$ accordingly into the bound of $\pi(y,y,y)$ in \eqref{eq:psi_n_bound} and multiply by the number of terms in the summation, i.e. $\big{\lceil}\frac{Q_j}{2}\big{\rceil}\cdot d$ or $\big{\lceil}\frac{R_j}{2}\big{\rceil}\cdot d$. This product is then repeated $q_j$ or $r_j$ times accordingly until the last level of the split, $j=\lceil\log_2s\rceil-1$, where we have $q_{\lceil\log_2s\rceil-1}$ and $Q_{\lceil\log_2s\rceil-1}$ (which is equal to 2). We exclude $R_{\lceil\log_2s\rceil-1}$ since $l_{R_{\lceil\log_2s\rceil-1}}=d$. Putting the whole product together results to \eqref{eq:prob_1set_ssparse_poly1} hence concluding the proof of Proposition \ref{pro:poly_bound}.
\end{proof}

As a final step we need the following corollary.
\begin{corollary}
\label{cor:poly_bound}
\begin{equation}
\label{eq:cor_poly_bound}
\Pi\left(l_s,\ldots,l_2,d\right) < \frac{2}{25\sqrt{2\pi s^3d^3}} \cdot \exp\left[3s\log(5d)\right].
\end{equation}
\end{corollary}

\begin{proof}
From Lemma \ref{lem:num_size_split} we can upper bound $R_j$ by $Q_j$ . Consequently \eqref{eq:prob_1set_ssparse_poly1} is upper bounded by the following.
\begin{multline}
\label{eq:prob_1set_ssparse_poly2}
\left(\bigg{\lceil}\frac{Q_0}{2}\bigg{\rceil}d\left(\frac{5}{4}\right)^2 \sqrt{2\pi\bigg{\lfloor}\frac{Q_0}{2}\bigg{\rfloor}d}\right)^{q_0} \times \\ \prod_{j=1}^{\lceil\log_2s\rceil-2} \left(\bigg{\lceil}\frac{Q_j}{2}\bigg{\rceil}d\left(\frac{5}{4}\right)^2 \sqrt{2\pi\bigg{\lfloor}\frac{Q_j}{2}\bigg{\rfloor}d}\right)^{q_j+r_j} \times \\
\left(\bigg{\lceil}\frac{Q_{\lceil\log_2s\rceil-1}}{2}\bigg{\rceil}d\left(\frac{5}{4}\right)^2 \sqrt{2\pi\bigg{\lfloor}\frac{Q_{\lceil\log_2s\rceil-1}}{2}\bigg{\rfloor}d}\right)^{q_{\lceil\log_2s\rceil-1}}
\end{multline}
Now we use the property that $q_j+r_j = 2^j$ for $j = 1,\ldots,\lceil\log_2s\rceil-1$ from Lemma \ref{lem:num_size_split} to bound \eqref{eq:prob_1set_ssparse_poly2} by the following.
\begin{equation}
\label{eq:prob_1set_ssparse_poly3}
\prod_{j=0}^{\lceil\log_2s\rceil-1} \left(\bigg{\lceil}\frac{Q_j}{2}\bigg{\rceil}d\left(\frac{5}{4}\right)^2 \sqrt{2\pi\bigg{\lfloor}\frac{Q_j}{2}\bigg{\rfloor}d}\right)^{2^j}.
\end{equation}
We have a strict upper bound when $r_{\lceil\log_2s\rceil-1} \neq 0$, which occurs when $s$ is not a power of $2$, because then by $q_j+r_j = 2^j$ we have $q_{\lceil\log_2s\rceil-1} + r_{\lceil\log_2s\rceil-1} = 2^{\lceil\log_2s\rceil-1}$. In fact \eqref{eq:prob_1set_ssparse_poly3} is an overestimate for a large $s$ which is not a power of $2$.

Note $Q_j = \big{\lceil}\frac{s}{2^j}\big{\rceil}$ by Lemma \ref{lem:num_size_split}. Thus $\Big{\lceil} \frac{Q_j}{2} \Big{\rceil} = \big{\lceil}\frac{s}{2^{j+1}}\big{\rceil}$ and $\Big{\lfloor} \frac{Q_j}{2} \Big{\rfloor} \leq \big{\lceil}\frac{s}{2^{j+1}}\big{\rceil}$. So we bound \eqref{eq:prob_1set_ssparse_poly3} by the following.
\begin{equation}
\label{eq:prob_1set_ssparse_poly4}
\prod_{j=0}^{\lceil\log_2s\rceil-1} \left(\Big{\lceil}\frac{s}{2^{j+1}}\Big{\rceil}d\left(\frac{5}{4}\right)^2 \sqrt{2\pi\Big{\lceil}\frac{s}{2^{j+1}}\Big{\rceil}d}\right)^{2^j}
\end{equation}
Next we upper bound $\lceil\log_2s\rceil-1$ in the limit of the product by $\log_2s$ and upper bound $\lceil\frac{s}{2^{j+1}}\rceil$ by $\frac{s}{2^{j+1}} + \frac{1}{2} = \frac{s}{2^{j+1}}\left(1+\frac{2^{j+1}}{s}\right)$, we also move the $d$ into the square root and combined the constants to have the following bound on \eqref{eq:prob_1set_ssparse_poly4}.
\begin{multline}
\prod_{j=0}^{\log_2s} \left[\frac{s}{2^{j+1}}\left(1+\frac{2^{j+1}}{s}\right)\left(\frac{25\sqrt{2\pi}}{16}\right) \right. \times \nonumber\\ \left. \sqrt{\frac{s}{2^{j+1}}\left(1+\frac{2^{j+1}}{s}\right)d^3}\right]^{2^j}.
\end{multline}
We bound $\left(1+\frac{2^{j+1}}{s}\right)$ by $2$ to bound the above by
\begin{multline}
\prod_{j=0}^{\log_2s} \left[\frac{s}{2^{j}}\left(\frac{25\sqrt{2\pi}}{16}\right) \sqrt{\frac{s}{2^{j}}d^3}\right]^{2^j} \\
\label{eq:prob_1set_ssparse_poly6b}
= \prod_{j=0}^{\log_2s} \left[\left(\frac{25\sqrt{2\pi}}{16}\right) \sqrt{\frac{s^3d^3}{2^{3j}}}\right]^{2^{j}}
\end{multline}
where we moved $s/2^j$ into the square root. Using the rule of indices the product of the constant term is replaced by it's power to sum of the indices. We then rearranged to have the power $3/2$ in the outside and this gives the following.
\begin{align}
\label{eq:prob_1set_ssparse_poly6c}
& \left(\frac{25\sqrt{2\pi}}{16}\right)^{\sum_{i=0}^{\log_2s}2^i} \left[\prod_{j=0}^{\log_2s}  \left(\frac{sd}{2^{j}}\right)^{2^j}\right]^{3/2}\\
\label{eq:prob_1set_ssparse_poly6d}
& = \left(\frac{25\sqrt{2\pi}}{16}\right)^{2s-1} \left[\left(sd\right)^{\sum_{i=0}^{\log_2s}2^i}\prod_{j=0}^{\log_2s}  \left(\frac{1}{2^{j}}\right)^{2^j}\right]^{3/2}\\
\label{eq:prob_1set_ssparse_poly6e}
& = \left(\frac{25\sqrt{2\pi}}{16}\right)^{2s-1} \left[\left(sd\right)^{2s-1}  \left(\frac{1}{2}\right)^{\sum_{j=0}^{\log_2s}j2^j}\right]^{3/2}.
\end{align}
From \eqref{eq:prob_1set_ssparse_poly6c} to \eqref{eq:prob_1set_ssparse_poly6d} we evaluate the power of the first factor which is a geometric series and we again use the rule of indices for the $sd$ factor. Then from \eqref{eq:prob_1set_ssparse_poly6d} to \eqref{eq:prob_1set_ssparse_poly6e} we use the indices' rule for the last factor and evaluate the power of the $sd$ factor which is also a geometric series. We simplify the power of the last factor by using the following.
\begin{equation}
 \label{eq:power_series}
\sum_{k=1}^{m}k\cdot2^k = (m-1)\cdot 2^{m+1} + 2.
\end{equation}
This therefore simplifies \eqref{eq:prob_1set_ssparse_poly6e} as follows.
\begin{align}
\label{eq:prob_1set_ssparse_poly7}
& \left(\frac{25\sqrt{2\pi}}{16}\right)^{2s-1} \left[\left(sd\right)^{2s-1}  \left(\frac{1}{2}\right)^{(\log_2s-1)\cdot 2^{\log_2s+1} + 2}\right]^{3/2}\\
\label{eq:prob_1set_ssparse_poly8}
& = \left(\frac{25\sqrt{2\pi}}{16}\right)^{2s-1} \left[\left(sd\right)^{2s-1}  \left(\frac{1}{2}\right)^{2s(\log_2s-1)}\frac{1}{4}\right]^{3/2}\\
\label{eq:prob_1set_ssparse_poly9}
& = \left(\frac{25\sqrt{2\pi}}{16}\right)^{2s-1} \left[\frac{(sd)^{2s}}{4sd}  2^{-2s\log_2s} 2^{2s}\right]^{3/2}\\
\label{eq:prob_1set_ssparse_poly10}
& = \left(\frac{25\sqrt{2\pi}}{16}\right)^{2s} \left(\frac{16}{25\sqrt{2\pi}}\right) \left[\frac{(2sd)^{2s}}{4sd}  s^{-2s} \right]^{3/2}\\
\label{eq:prob_1set_ssparse_poly11}
& = \left(\frac{25\sqrt{2\pi}}{16}\right)^{2s} \left(\frac{16}{25\sqrt{2\pi}}\right) \left[\frac{(2d)^{2s}}{4sd}\right]^{3/2}.
\end{align}
From \eqref{eq:prob_1set_ssparse_poly7} through \eqref{eq:prob_1set_ssparse_poly9} we simplified using basic properties of indices and logarithms. While from \eqref{eq:prob_1set_ssparse_poly9} to \eqref{eq:prob_1set_ssparse_poly10} we incorporated $2^{2s}$ into the first factor inside the square brackets and we rewrote the first factor into a product of a power in $s$ and another without $s$. From \eqref{eq:prob_1set_ssparse_poly10} to \eqref{eq:prob_1set_ssparse_poly11} the $s^{2s}$ and $s^{-2s}$ canceled out.

\vspace{3mm}

Now we expand the square brackets in \eqref{eq:prob_1set_ssparse_poly11} to have \eqref{eq:prob_1set_ssparse_poly12} below.
\vspace{-2mm}

\begin{align}
\label{eq:prob_1set_ssparse_poly12}
& \left(\frac{25\sqrt{2\pi}}{16}\right)^{2s} \left(\frac{16}{25\sqrt{2\pi}}\right) \frac{1}{8\sqrt{s^3d^3}} (2d)^{3s}\\
\label{eq:prob_1set_ssparse_poly13}
& = \left(\frac{25\sqrt{2\pi}}{16}\right)^{2s}  (2d)^{3s} \frac{2}{25\sqrt{2\pi s^3d^3}}\\
\label{eq:prob_1set_ssparse_poly14}
& = \frac{2}{25\sqrt{2\pi s^3d^3}} \cdot \exp\left(3s\log\left(2\left(\frac{25\sqrt{2\pi}}{16}\right)^{2/3}d\right)\right)\\
\label{eq:prob_1set_ssparse_poly15}
& <  \frac{2}{25\sqrt{2\pi s^3d^3}} \cdot \exp\left[3s\log(5d)\right]
\end{align}
From \eqref{eq:prob_1set_ssparse_poly12} to \eqref{eq:prob_1set_ssparse_poly13} we simplified and from \eqref{eq:prob_1set_ssparse_poly13} to \eqref{eq:prob_1set_ssparse_poly14} we rewrote the powers as an exponential with a logarithmic exponent. Then from \eqref{eq:prob_1set_ssparse_poly14} to \eqref{eq:prob_1set_ssparse_poly15} we upper bounded $2\left(\frac{25\sqrt{2\pi}}{16}\right)^{2/3}$ by $5$ which gives the required format of a product of a polynomial and an exponential to conclude the proof.
\end{proof}

With the bound in Corollary \ref{cor:poly_bound} we have completed the bounding of $\Pi\left(l_s,\ldots,l_2,d\right)$ in \eqref{eq:prob_1set_ssparse8}. Next we bound $\Psi_n\left(a_s,\ldots,a_2,d\right)$ which is given by \eqref{eq:prob_1set_ssparse_exp1}. Lemma \ref{lem:num_size_split} gives the three arguments for each $\psi_n(\cdot)$ and the number of $\psi_n(\cdot)$ with the same arguments. Using this lemma we express $\Psi_n\left(a_s,\ldots,a_2,d\right)$ as
\begin{multline}
    \label{eq:Psi_n}
    \sum_{j=0}^{\lceil\log_2(s)\rceil-2} \left[ q_j \cdot \psi_n\left(a_{Q_j},a_{\big{\lceil}\frac{Q_j}{2}\big{\rceil}},a_{\big{\lfloor}\frac{Q_j}{2}\big{\rfloor}}\right) + \right.\\
    \left. r_j \cdot \psi_n\left(a_{R_j},a_{\big{\lceil}\frac{R_j}{2}\big{\rceil}},a_{\big{\lfloor}\frac{R_j}{2}\big{\rfloor}}\right) \right] + \\
    q_{\lceil\log_2(s)\rceil-1} \cdot \psi_n\left(a_{2},d,d\right).
\end{multline}
Equation \eqref{eq:Psi_n} is bounded above in Lemma \ref{lem:Psi_n_bound} by the following.
\begin{equation}
    \label{eq:Psi_n_bound2a}
     \sum_{j=0}^{\lceil\log_2(s)\rceil-1} 2^j \cdot \psi_n\left(a_{Q_j},a_{\big{\lfloor}\frac{R_j}{2}\big{\rfloor}},a_{\big{\lfloor}\frac{R_j}{2}\big{\rfloor}}\right).
\end{equation}
If we let the $a_{2i} = a_{Q_j}$ and $a_i = a_{\big{\lfloor}\frac{R_j}{2}\big{\rfloor}}$ we have \eqref{eq:Psi_n_bound2a} equal to the following.
\begin{multline}
\label{eq:Psi_n_bound2b}
\sum_{i=1}^{\lceil s/2\rceil} \frac{s}{2i} \cdot \psi_n\left(a_{2i},a_i,a_i\right) = \sum_{i=1}^{\lceil s/2\rceil} \frac{s}{2i}\left[a_i\cdot \hbox{H}\left(\frac{a_{2i}-a_i}{a_i}\right)\right. \\
\left. + \left(n-a_i\right) \cdot \hbox{H}\left(\frac{a_{2i}-a_i}{n-a_i}\right) - n\cdot \hbox{H}\left(\frac{a_i}{n}\right) \right].
\end{multline}
Now we combine the bound of $\Pi\left(l_s,\ldots,l_2,d\right)$ in \eqref{eq:cor_poly_bound} and the exponential whose exponent is the bound of $\Psi_n\left(a_s,\ldots,a_2,d\right)$ in \eqref{eq:Psi_n_bound2b} to get \eqref{eq:pmax}, the polynomial $p_{max}(s,d) = \frac{2}{25\sqrt{2\pi s^3d^3}} $, and \eqref{eq:Psi}, the exponent of the exponential $\Psi\left(a_s,\ldots,d\right)$ given by the sum of $3s\log\left(5d \right)$ and the right hand side of \eqref{eq:Psi_n_bound2b}.

Lemma \ref{lem:max_Psi_n_bound} gives the $a_i$ that maximize \eqref{eq:Psi_n_bound2b} and the systems \eqref{eq:prob_1set_ssparse_expo_bound_grad_cp1} and \eqref{eq:prob_1set_ssparse_expo_bound_grad_cp2} they satisfy depending on the constraints on $a_s$. Solving completely the system \eqref{eq:prob_1set_ssparse_expo_bound_grad_cp1} gives $\hat{a}_i$ in \eqref{eq:prob_1set_ssparse_expo_bound_grad_cp3} and \eqref{eq:aj_system_unconstrained} which are the expected values of the $a_i$. The system \eqref{eq:aj_system_constrained} is equivalent to \eqref{eq:prob_1set_ssparse_expo_bound_grad_cp2} hence also proven in Lemma \ref{lem:max_Psi_n_bound}. This therefore concludes the proof Theorem \ref{thm:prob_bound_1set_expansion}.

\subsection{Main Corollaries}\label{sec:main_corollaries}
In this section we present the proofs of the corollaries in Sections \ref{sec:main} and \ref{sec:ric1}. These include the proof of Corollary \ref{cor:prob_bound_1set_expander} in Section \ref{sec:main_corollary_proof}, the proof of Corollary \ref{cor:prob_expander_existence1} in Section \ref{sec:cor_prob_expander_exist1} and the proof of Corollary \ref{cor:prob_expander_existence2} given in Section \ref{sec:cor_prob_expander_exist2}.

\subsubsection{Corollary \ref{cor:prob_bound_1set_expander}}\label{sec:main_corollary_proof}

Satisfying RIP-1 means that for any $s-$sparse vector $x$, $\|A_Sx\|_1 \geq (1-2\e)d\|x\|_1$ which indicates that the cardinality of the set of neighbors satisfies $|A_s| \geq (1-\e)ds$. Therefore
\begin{multline}
\hbox{Prob}\left(\|A_Sx\|_1 \leq (1-2\e)d\|x\|_1\right) \\ \equiv \hbox{Prob}\left(|A_s| \leq (1-\e)ds\right).
\end{multline}
This implies that $a_s = (1-\e)ds$ and since this is restricting $a_s$ to be less than it's expected value given by \eqref{eq:aj_system_unconstrained}, the rest of the $a_i$ satisfy the polynomial system \eqref{eq:aj_system_constrained}. If there exists a solution then the $a_i$ would be functions of $s,~d$ and $\e$ which makes $\Psi\left(a_s,\ldots,a_2,d\right) = \Psi\left(s,d,\e\right)$.

\subsubsection{Corollary \ref{cor:prob_expander_existence1}} \label{sec:cor_prob_expander_exist1}
Corollary \ref{cor:prob_bound_1set_expander} states that by fixing $S$ and the other parameters, $\hbox{Prob}\left(\mathop{\|A_Sx\|_1} \leq (1-2\e)d\|x\|_1\right) < p_{max}(s,d)\cdot \exp\left[n\cdot\Psi\left(s,d,\e\right)\right]$. Corollary \ref{cor:prob_expander_existence1} considers any $S\subset [N]$ and since the matrices are adjacency matrices of lossless expanders we need to consider any $S\subset [N]$ such that $|S|\leq k$. Therefore our target is $\hbox{Prob}\left(\mathop{\|Ax\|_1} \leq (1-2\e)d\|x\|_1\right)$ which is bounded by a simple union bound over all $\binom{N}{s} ~S$ sets and by treating each set $S$, of cardinality less than $k$, independent we sum over this probability to get the following bound.
\begin{align}
\label{eq:pro_expander_exist1a}
& \sum_{s=2}^{k} \binom{N}{s} \cdot \hbox{Prob}\left(\mathop{\|A_Sx\|_1} \leq (1-2\e)d\|x\|_1\right)\\
\label{eq:pro_expander_exist1b}
& < \sum_{s=2}^{k} \binom{N}{s} \cdot p_{max}(s,d) \cdot \exp\left[n\cdot\Psi\left(s,d,\e\right)\right]\\
\label{eq:pro_expander_exist1c}
& < \sum_{s=2}^{k} \left(\frac{5}{4}\right)^2 \frac{1}{\sqrt{2\pi s\left(1-\frac{s}{N}\right)}} \cdot p_{max}(s,d) \nonumber \\ & \quad \times \exp\left[N\hbox{H}\left(\frac{s}{N}\right) + n\cdot\Psi\left(s,d,\e\right)\right]\\
\label{eq:pro_expander_exist1d}
& < k\left(\frac{5}{4}\right)^2 \frac{p_{max}(k,d)}{\sqrt{2\pi k\left(1-\frac{k}{N}\right)}} \nonumber\\
& \quad \times \exp\left[N\left(\hbox{H}\left(\frac{k}{N}\right) + \frac{n}{N}\cdot\Psi\left(
k,d,\e\right)\right)\right].
\end{align}
From \eqref{eq:pro_expander_exist1a} to \eqref{eq:pro_expander_exist1b} we bound the probability in \eqref{eq:pro_expander_exist1a} using Corollary \ref{cor:prob_bound_1set_expander}. Then from \eqref{eq:pro_expander_exist1b} to \eqref{eq:pro_expander_exist1c} we bound $\binom{N}{s}$ using Stirling's formula \eqref{eq:stirling} by a polynomial in $N$ multiplying $p_{max}(s,d)$ and an exponential incorporated into the exponent of the exponential term. From \eqref{eq:pro_expander_exist1c} to \eqref{eq:pro_expander_exist1d} we use that for $N>2k$ the entropy $\hbox{H}\left(\frac{s}{N}\right)$ is largest when $s=k$ and we bound the summation by taking the maximum value of $s$ and multiplying by the number of terms plus one, giving $k$, in the summation. This gives $p'_{max}(N,k,d) = k\left(\frac{5}{4}\right)^2 \frac{p_{max}(k,d)}{\sqrt{2\pi k\left(1-\frac{k}{N}\right)}}$ which simplifies to $\frac{1}{16\pi k\sqrt{d^3\left(1-\frac{k}{N}\right)}}$ and the factor $\Psi_{net}\left(k,n,N;d,\e\right) = \hbox{H}\left(\frac{k}{N}\right) + \frac{n}{N}\cdot\Psi\left(k,d,\e\right)$ is what is multiplied to $N$ in the exponent as claimed.

\subsubsection{Corollary \ref{cor:prob_expander_existence2}}  \label{sec:cor_prob_expander_exist2}

Corollary \ref{cor:prob_expander_existence1} has given us an upper bound on the probability $\hbox{Prob}\left(\|Ax\|_1 \leq (1-2\e)d\|x\|_1\right)$ in \eqref{eq:prob_expander_existence1}. In this bound the exponential dominates the polynomial. Consequently, in the limit as $(k,n,N) \rightarrow \infty$ while $k/n \rightarrow \r \in (0,1)$ and $n/N \rightarrow \d \in (0,1)$ this bound has a sharp transition at the zero level curve of $\Psi_{net}$.  For $\Psi_{net}\left(k,n,N;d,\e\right)$ strictly bounded above zero the overall bound grows exponentially in $N$ without limit, while for $\Psi_{net}\left(k,n,N;d,\e\right)$ strictly bounded below zero the overall bound decays to zero exponentially quickly.  We define $\r^{exp}(\d;d,\e)$ to satisfy $\Psi_{net}\left(k,n,N;d,\e\right)=0$ in \eqref{eq:rho_exp}, so that for any $\r$ strictly less than $\r^{exp}(\d;d,\e)$ the exponent will satisfy $\Psi_{net}\left(k,n,N;d,\e\right)<0$ and hence the bound decay to zero.

More precisely, for $k/n \rightarrow \r < (1-\g)\r^{exp}(\d;d,\e)$ with small $\g>0$, in this regime of $\r < (1-\g)\r^{exp}(\d;d,\e)$ we have $\hbox{Prob}\left(\|Ax\|_1 \leq (1-2\e)d\|x\|_1\right) \rightarrow 0$. Therefore, $\hbox{Prob}\left(\|Ax\|_1 \geq (1-2\e)d\|x\|_1\right) \rightarrow 1$ as the problem size grows such that $(k,n,N) \rightarrow \infty$, $n/N \rightarrow \d \in (0,1)$ and $k/n \rightarrow \r$.

\section{Appendix}\label{sec:appendix}

\subsection{Proof of Corollary \ref{cor:prob_expander_existenceBI}}\label{sec:prf:expander_existence}

The first part of this proof uses ideas from the proof of Proposition \ref{pro:para_opt} which is the same as Theorem 16 in \cite{berinde2009advances}. We consider a bipartite graph $G=(U,V,E)$ with $|U|=N$ left vertices, $|V|=n$ right vertices and left degree $d$. For a fixed $S\subset U$ where $|S|=s\leq k$, $G$ fails to be an expander on $S$ if $|\Gamma(S)|<(1-\e)ds$. This means that in a sequence of $ds$ vertex indices at least $\e ds$ of the these indices are in the collision set that is identical to some preceding value in the sequence.

Therefore, the probability that a neighbor chosen uniformly at random is to be in the collision set is at most $ds/n$ and, treating each event independently, then the probability that a set of $\e ds$ neighbors chosen at random are in the collision set is at most $\left(ds/n\right)^{\e ds}$. There are $\binom{ds}{\e ds}$ ways of choosing a set of $\e ds$ points from a set of $ds$ points and $\binom{N}{s}$ ways of choosing each set $S$ from $U$. This means therefore that the probability that $G$ fails to expand in at least one of the sets $S$ of fixed size $s$ can be bounded above by a union bound
\begin{multline}
\hbox{Prob}\left(G ~\mbox{fails to expand on} ~S\right) \\ \leq \binom{N}{s}\binom{ds}{\e ds}\left(\frac{ds}{n}\right)^{\e ds}.
\label{eq:prob_Gfails_S}
\end{multline}

We define $p_s$ to be the right hand side of \eqref{eq:prob_Gfails_S} and we use the right hand side of the Stirling's inequality \eqref{eq:stirling} to upper bound $p_s$ as thus
\begin{multline}
\label{eq:prob_Gfails_S2}
p_s < \frac{5}{4}\left[2\pi\frac{\e ds}{ds}\left(1-\frac{\e ds}{ds}\right)\e ds\right]^{-\frac{1}{2}} \exp\left[ds\mbox{H}\left(\frac{\e ds}{ds}\right)\right] \\ \times \frac{5}{4}\left[2\pi\frac{s}{N}(1-\frac{s}{N})N\right]^{-\frac{1}{2}} \\ \times \exp\left[N\mbox{H}\left(\frac{s}{N}\right)\right] \times \left(\frac{ds}{n}\right)^{\e ds}
\end{multline}
Writing the last multiplicand of \eqref{eq:prob_Gfails_S2} in exponential form and simplifying the expression gives

\begin{equation}
\label{eq:prob_Gfails_S3a}
 p_s < p_{max}(N,s;d,\e) \cdot \exp\left[N\cdot\Psi\left(s,n,N;d,\e\right)\right],
\end{equation}
where $\Psi\left(s,n,N;d,\e\right)$ is
\begin{equation}
\label{eq:psiBI}
\mbox{H}\left(\frac{s}{N}\right) + \frac{ds}{N}\mbox{H}\left(\e\right) + \frac{\e ds}{N} \log\left(\frac{ds}{n}\right),
\end{equation}
and $p_{max}(N,s;d,\e)$ is a polynomial in $N$ and $s$ for each $d$ and $\e$ fixed given by
\begin{equation}
\label{eq:pmaxBI}
\left(\frac{5}{4}\right)^2 \cdot \frac{1}{2\pi s} \cdot \left[\frac{N}{\e(1-\e)(N-s)d}\right]^{\frac{1}{2}}.
\end{equation}

Finally $G$ fails to be an expander if it fails to expand on at least one set $S$ of any size $s\leq k$. This means therefore that
\begin{equation}
\hbox{Prob}\left(G ~\mbox{fails to be an expander}\right) \leq \sum_{s=1}^k{p_s}.
\label{eq:prob_Gfails}
\end{equation}
From \eqref{eq:prob_Gfails_S3a} we have $\sum_{s=2}^k{p_s}$ bounded by
\begin{align}
\label{eq:prob_Gfails_2a}
 & \sum_{s=2}^k p_{max}(N,s;d,\e) \cdot \exp\left[N\cdot\Psi\left(s,n,N;d,\e\right)\right]\\
\label{eq:prob_Gfails_2b}
    & < p'_{max}(N,k;d,\e) \cdot \exp\left[N\cdot\Psi\left(k,n,N;d,\e\right)\right],
\end{align}
where $p'_{max}(N,k;d,\e) = k\cdot p_{max}(N,k;d,\e)$ and we achieved the bound from \eqref{eq:prob_Gfails_2a} to \eqref{eq:prob_Gfails_2b} by upper bounding the sum with the product of the largest term in the sum (which is when $s=k$ since $k<N/2$) and one plus the number of terms in the sum, giving $k$. Hence from \eqref{eq:prob_Gfails} and \eqref{eq:prob_Gfails_2b} we have
\begin{multline}
\hbox{Prob}\left(G ~\mbox{fails to be an expander}\right) < p'_{max}(N,k;d,\e) \\ \times  \exp\left[N\cdot\Psi\left(k,n,N;d,\e\right)\right].
\label{eq:prob_Gfails3}
\end{multline}

As the problem size, $(k,n,N)$, grows the exponential term will be driving the probability in \eqref{eq:prob_Gfails3}, hence having
\begin{equation}
\label{eq:psi_condition}
\Psi\left(k,n,N;d,\e\right)<0
\end{equation}
yields $\hbox{Prob}\left(G ~\mbox{fails to be an expander}\right) \rightarrow 0$ as the problem size $(k,n,N)\rightarrow \infty$.

Let $k/n \rightarrow \r\in(0,1)$ and $n/N \rightarrow \d\in(0,1)$ as $(k,n,N) \rightarrow \infty$ and we define  $\r^{exp}_{bi}\left(\d;d,\e\right)$ as the limiting value of $k/n$ that satisfies $\Psi\left(k,n,N;d,\e\right)=0$ for each fixed $\e$ and $d$ and all $\d$. Note that for fixed $\e,~d$ and $\d$ it is deducible from our analysis of $\psi_n(\cdot)$ in Section \ref{sec:tools} that $\Psi\left(k,n,N;d,\e\right)$ is a strictly monotonically increasing function of $k/n$. Therefore for any $\r<\r^{exp}_{bi}, ~\Psi\left(k,n,N;d,\e\right) < 0$ as $(k,n,N) \rightarrow \infty, ~\hbox{Prob}\left(G ~\mbox{fails to be an expander}\right) \rightarrow 0$ and $G$ becomes an expander with probability approaching one exponentially in $N$ which is the same as exponential growth in $n$ since $n\rightarrow N\r$.



\begin{IEEEbiographynophoto}{Bubacarr Bah} received the BSc degree in mathematics and physics from the University of The Gambia, and the MSc degree in mathematical modeling and scientific computing from the University of Oxford, Wolfson College. He received his PhD degree in applied and computational mathematics at the University of Edinburgh where his PhD degree supervisor was Jared Tanner. 

He is currently a POSTDOCTORAL FELLOW at EPFL working with Volkan Cevher. Previously he had been a GRADUATE ASSISTANT at the University of The Gambia (2004-2007). His research interests include random matrix theory with applications to compressed sensing and sparse approximation. 

Dr. Bah is a member of SIAM and has received the SIAM Best Student Paper Prize (2010).
\end{IEEEbiographynophoto}

\begin{IEEEbiographynophoto}{Jared ~Tanner}
received the B.S. degree in physics from the University of Utah, and the Ph.D. degree in Applied Mathematics from UCLA where his PhD degree adviser was Eitan Tadmor. 

He is currently the PROFESSOR of the mathematics of information at the University of Oxford and Fellow of Exeter College.  Previously he had the academic posts: PROFESSOR of the mathematics of information at the University of Edinburgh (2007-2012), ASSISTANT PROFESSOR at the University of Utah (2006-2007), National Science Foundation POSTDOCTORAL FELLOW in the mathematical sciences at Stanford University (2004-2006), and a VISITING ASSISTANT RESEARCH PROFESSOR at the University of California at Davis (2002-2004). His research interests include signal processing with emphasis in compressed sensing, sparse approximation, applied Fourier analysis, and spectral methods. 

Prof. Tanner has received the Philip Leverhulme Prize (2009), Sloan Research Fellow in Science and Technology (2007), Monroe Martin Prize (2005), and the Leslie Fox Prize (2003).
\end{IEEEbiographynophoto}

\end{document}